\documentclass[journal]{IEEEtran}
\usepackage{placeins}
\usepackage{diagbox}
\usepackage[inline]{enumitem}
\usepackage{xcolor}
\usepackage{amsmath,amsfonts}
\usepackage{algorithm}
\usepackage[group-separator={,},group-minimum-digits=4]{siunitx}
\usepackage{threeparttable}
\usepackage{array}
\usepackage[caption=false,font=normalsize,labelfont=sf,textfont=sf]{subfig}
\usepackage{textcomp}
\usepackage{stfloats}
\usepackage{url}
\usepackage{verbatim}
\usepackage{graphicx}
\usepackage{algpseudocode}
\usepackage{cite}
\usepackage{balance}
\usepackage{mathrsfs}
\usepackage[colorlinks=true,
            linkcolor=blue,
            anchorcolor=blue,
            citecolor=blue]{hyperref}

\begin{document}

\title{
     \fontsize{16.5pt}{\baselineskip}\selectfont{Bistatic Integrated Sensing and Communication in the Presence of a Disco Reconfigurable Intelligent Surface: Disruption, Enhancement, or Both?}
    % \fontsize{18 pt}{\baselineskip}\selectfont{On the Performance of Bistatic Integrated Sensing and Communication in the Presence of a Disco Reconfigurable Intelligent Surface}
%On the Performance of Integrated Sensing and Communications Under DISCO Physical-Layer Jamming Attacks 
%Bistatic Integrated Sensing and Communication with Imperfect CSI Deliberately Induced by Random RIS Configurations
} % \fontsize{20.5pt}{\baselineskip}\selectfont
\author{ 
{%Authors       
    %
    % Luyao~Sun,
    % Yongxing~Song,
    Huan~Huang,~\textit{Member,~IEEE}, 
    Hongliang~Zhang,~\textit{Member,~IEEE},
    Weidong~Mei,~\textit{Member,~IEEE},  
    Minghui~Min,~\textit{Senior~Member,~IEEE}, 
    %Dongdong~Zou,
    %Yi~Cai, %~\textit{Senior~Member,~IEEE},
    % Gangxiang~Shen,~\textit{Senior~Member,~IEEE},\\
    %Lingyang~Song,~\textit{Fellow,~IEEE},
    %Dusit~Niyato,~\textit{Fellow,~IEEE},
    %A.~Lee~Swindlehurst,~\textit{Fellow,~IEEE},
    and~Zhu~Han,~\textit{Fellow,~IEEE}
}% <-this % stops a space
\thanks{  
% This work was supported by  
% the National Natural Science Foundation of China (62401387, 62371011), 
% the Natural Science Foundation of Jiangsu Province (BK20240768),
% and partially supported by 
% the U.S. National Science Foundation (CNS--2107216, CNS--2128368, CNS--2107182, CMMI--2222810, ECCS-2302469, ECCS--2030029), 
% US Department of Transportation, Toyota and Amazon. (\textit{Corresponding author: Huan~Huang}).

%A portion of this work was published in~\cite{MyDISCOISACWCL}.
H.~Huang is with the School of Electronic and Information Engineering, Soochow University, Suzhou, Jiangsu 215006, China (e-mail: hhuang1799@gmail.com). %\{jide\_yuan, ljun\}@suda.edu.cn ylzheng@suda.edu.cn,  , D.~Zou, Y.~Cai, and G~Shen are ddzou@suda.edu.cn, yicai@ieee.org, shengx@suda.edu.cn lysun02@163.com; 15564267992@163.com; hhuang1799@gmail.com; yicai@ieee.org
 
H.~Zhang  is with the State Key Laboratory of Advanced Optical Communication Systems and Networks, School of Electronics, Peking University, Beijing 100871, China (email: hongliang.zhang92@gmail.com). 
%,State Key Laboratory of Advanced Optical Communication Systems and Networks,   boya.di@pku.edu.cn
%H.~Zhang is with the School of Electronics, Peking University, Beijing 100871, China (e-mail: hongliang.zhang92@gmail.com)  lingyang.song@pku.edu.cn

W.~Mei is with the National Key Laboratory of Wireless Communications, University of Electronic Science and Technology of China, Chengdu 611731, China  (e-mail: wmei@uestc.edu.cn).

M.~Min is with the School of Information and Control Engineering, 
China University of Mining and Technology, Xuzhou 221116, China (minmh@cumt.edu.cn).

% R.~Zhang is with the School of Science and Engineering, Shenzhen Research Institute of Big Data, Chinese University of Hong Kong, Shenzhen 518172, China, 
% and also with the Department of Electrical and Computer Engineering, National University of Singapore, Singapore 117583 (e-mail: rzhang@cuhk.edu.cn).
 
% D.~Niyato is with the College of Computing and Data Science, Nanyang Technological University, Singapore 639798 (e-mail: dniyato@ntu.edu.sg).

% A.~L.~Swindlehurst is with the Center for Pervasive Communications and Computing, University of California, Irvine, CA 92697, USA (e-mail: swindle@uci.edu).

Z.~Han is with the Department of Electrical and Computer Engineering at the University of Houston, Houston, TX 77004 USA  
%and also with the Department of Computer Science and Engineering, Kyung Hee University, Seoul, South Korea, 446-701.  
(email: hanzhu22@gmail.com).
}
}
\maketitle

\begin{abstract}
    Integrated sensing and communication (ISAC) is widely regarded as 
    one of the key enabling technologies for the future sixth generation (6G) 
    wireless communication systems. In this work, we investigate a bistatic ISAC system in the presence of a disco reconfigurable intelligent surface (DRIS), whose random and time-varying reflection coefficients emulate a ``disco ball''. The introduction of DRIS breaks the underlying assumption in existing ISAC systems that the sensing and communication channels remain static or  quasi-static within the channel coherence time.
    We first develop a bistatic system model 
    incorporating the DRIS and characterize all involved wireless channels.
    Then, an ISAC waveform design that balances sensing and communication performance is proposed by 
    formulating a Pareto optimization problem,
    where the trade-off is controlled through a tunable factor.
    Communication and sensing performance in the bistatic ISAC system 
    are quantified by the signal-to-interference-plus-noise ratio (SINR) 
    and the Cram$\acute{\rm e}$r-Rao lower bound (CRLB), respectively.
    To quantify the impact of DRIS on the bistatic ISAC system,
    we derive the statistical characteristics  of DRIS-induced active channel aging (ACA) channels for communications 
    and the cascaded DRIS-based sensing channel.
    Then, we establish a theoretical lower bound on the SINR and closed-form CRLB expressions in the presence of a DRIS.
    The analysis reveals several distinctive properties of DRIS in bistatic ISAC systems. 
    In particular, the DRIS degrades communication performance significantly
    due to the introduction of ACA interference.
    In contrast, with respect to sensing performance, 
    the DRIS decreases the estimation accuracy of the angle of departure (AoD) 
    while concurrently enhancing that of the angle of arrival (AoA).
    Numerical results validate the derived theoretical analysis and confirm these DRIS-induced behaviors.
\end{abstract}


\section{Introduction}\label{Intro}
Integrated sensing and communication (ISAC) has attracted increasing attention 
as a promising technology for future sixth-generation (6G) wireless systems~\cite{ISACMag,ISACiotj}. 
In particular, ISAC implements joint target sensing and data communication functionalities
using the same radio-frequency (RF) hardware and computing platform~\cite{HLZhangISAC,HeathISAC}, 
which offers an exciting opportunity to implement sensing using existing wireless communication infrastructure~\cite{LiuJSAC}.
The added sensing functionality enabled by the environmental awareness makes ISAC a fundamental component of future 6G smart environments. 
Therefore, both industry and academia have undertaken extensive research into ISAC-related issues, including ISAC waveform design~\cite{FLiuISAC,FLiuISAC1} 
and the theoretical analysis of ISAC waveform~\cite{MUIImpact,CRLBISAC},
and the sensing performance metrics, i.e., mean squared error (MSE) or the Cram$\acute{\rm e}$r-Rao lower bound (CRLB) are 
widely used~\cite{CRLBISACLS}.

To further enhance the performance of ISAC systems, reconfigurable intelligent surfaces (RISs) have been anticipated to play a critical role~\cite{AORIS,MyMultiRIS,RISExp}. These surfaces consist of a large number of reflective elements whose reflection coefficients can be adjusted using simple PIN or varactor diodes~\cite{IRSCuiTJ,IRSsur1,IRSsur2}. By properly configuring these coefficients, the electromagnetic environment can be reshaped to improve both sensing and communications~\cite{LSWCMISAC}.
For example, the RISs can provide extra wireless links when direct links between an ISAC base station (BS) 
and sensing and communication users are blocked~\cite{ExtraPathISAC,CRLBISACLS}. Moreover, the authors in~\cite{STARISAC} proposed a simultaneously transmitting and reflecting RIS-assisted ISAC scheme, where the simultaneously transmitting and reflecting RIS partitions the entire wireless space into a sensing space and a communication space, respectively.

However, the performance enhancement enabled by  RISs relies on the underlying assumption that 
the sensing and communication channels remain static or quasi-static within the channel coherence time. 
This is because the RIS coefficients and the ISAC waveform are designed based on the estimated channel state information (CSI).
While this underlying assumption is normally valid, it can be negated 
when a so-called disco RIS (DRIS) is deployed~\cite{MyWCMag}.
The concept of DRIS was first introduced in the context of a fully-passive jammer (FPJ)~\cite{MyTVT}, where the DRIS with time-varying and random reflection coefficients acts like a ``disco ball.''
Consequently, active channel aging (ACA) interference is introduced, which
causes the wireless channels to change more rapidly than their channel coherence time~\cite{DIRSTWC,TWCAnti,TWCDIOS}.
The ACA interference imposed by DRISs is different from 
the channel aging (CA) interference in traditional MU-MISO systems, 
which is caused by time variations in wireless propagation and delays in computation 
between the time the channels are learned and when they are used for precoding~\cite{ChanAge}.
% This type of ACA interference is referred to as DISCO jamming attacks.

Prior works have extensively investigated the impact of DRIS-based FPJs on pure communication systems.
In~\cite{DIRSTWC}, the authors quantified the performance loss due to the DRIS-based fully-passive jamming attacks by deriving the lower bound of ergodic achievable uplink rate.
Furthermore, they showed some interesting properties of this fully-passive jamming. 
For example, unlike the classical active jammer (AJ), which inflicts intentional interference to block the communication,
the ACA impact caused by the DRIS-based FPJ cannot be mitigated by classical anti-jamming approaches 
such as increasing the transmit power of the legitimate system. 
In~\cite{TWCAnti}, the authors provided that the fully-passive jamming launched by the DRIS-based FPJ
behaves like additive Gaussian white noise (AWGN) 
when the number of DRIS reflective elements is massive.
Based on this AWGN property, 
an anti-jamming scheme was designed  
by leveraging the statistical characteristics:
 the channel response of the cascaded DRIS-based channel 
is approximated as complex Gaussian~\cite{TWCAnti}.

In addition, some studies have investigated the use of DRISs to break key consistency~\cite{OtherDIRS1,OtherDIRS2},
where the channel reciprocity in time division duplex (TDD) systems 
is destroyed by the ACA  to make the channel reciprocity based key generation fail.
Consequently, the DRIS concept has been extended to beyond-diagonal RISs (BD-RISs)
by employing non-reciprocal connections between their elements~\cite{LeeDRIS},
and to omni-RISs by randomly adjusting their refractive and reflective coefficients~\cite{TWCDIOS}.
Moreover, the authors in~\cite{CCHuang} have 
shown that a DRIS deployed by the warden Willie improves his detection error probabilities 
while jamming covert transmissions between Alice and Bob.

Despite these advances, existing ISAC studies predominantly focus on performance improvement 
under the assumption of the channel reciprocity in TDD systems~\cite{LSWCMISAC,ExtraPathISAC,CRLBISACLS,STARISAC,RISISACadd1}, 
largely overlooking the potential threats posed by DRISs. 
The introduction of a DRIS breaks the  underlying assumption in ISAC systems, 
raising significant physical-layer security concerns.
Although~\cite{MyDISCOISACWCL} provided a preliminary illustration of the DRIS impact on a monostatic ISAC system, it has several limitations. 
First, it relies solely on Monte Carlo simulations to demonstrate the DRIS impact, lacking a rigorous theoretical quantification of how DRIS degrades system performance. 
Second, it is limited to a simple monostatic sensing scenario, leaving the more challenging and general case of bistatic sensing~\cite{BistaticRadar} unexplored. 
Consequently, it remains unclear whether and how the presence of a DRIS degrades the performance of bistatic ISAC systems, particularly regarding the sensing accuracy in terms of angle-of-departure (AoD) and angle-of-arrival (AoA) estimation.

Motivated by the above facts,  we investigate and quantify the impact of DRIS on bistatic ISAC systems in this work. 
The main contributions are summarized as follows:
\begin{itemize}
\item 
We first model a bistatic ISAC system in the presence of a DRIS,
where the time-varying reflection coefficients of the DRIS are randomly generated by an independent DRIS controller 
and are uncorrelated with the bistatic ISAC system. 
It is worth noting that 
 the DRIS can be implemented without relying on either CSI of wireless channels or additional transmit power.
We then characterize all involved wireless channels.
The DRIS introduces ACA into both bistatic sensing and communication channels, 
even within their channel coherence intervals.
As a result, the underlying assumption in bistatic ISAC systems 
that the sensing and communication channels remain static or quasi-static over 
the channel coherence time is violated.
\item  
Based on the built bistatic ISAC model, we design an ISAC waveform by formulating a Pareto optimization problem that explicitly balances sensing and communication performance.
By solving this problem for different trade-off factors, 
the resulting ISAC waveform achieves different trade-offs between the sensing and communication performance
under the same total transmit power.
For a given ISAC waveform, SINR is used to demonstrate the impact of DRIS-induced ACA interference on 
the communication performance of the bistatic ISAC system.
Moreover, we employ CRLB to characterize the influence of DRIS on the AoD and AoA estimates.
\item  
In order to analytically quantify the impact of DRIS on the bistatic ISAC system,
we first derive the statistical characteristics of the DRIS-induced ACA channel for communications 
and that of the cascaded DRIS-based sensing channel.
Based on the derived statistical characteristics,
we establish a theoretical lower bound of SINR and derive the closed-form  CRLB
for the AoD and AoA  estimates in the presence of a DRIS. 
Numerical results are provided to validate the tightness and accuracy of the derived SINR lower bound and the CRLB expressions.
\item
According to the derivations about SINR and CRLB,
we demonstrate some interesting properties of the DRIS impact on the bistatic ISAC system.
For instance,
the DRIS degrades communication performance significantly due to the
introduction of ACA interference.
In contrast, with respect to sensing performance, 
the DRIS decreases the estimation accuracy
of the AoD while improving that of the AoA.
These properties indicate that the introduction of DRIS in bistatic ISAC systems may simultaneously degrade communication reliability 
and reshape sensing accuracy in a direction-dependent manner.
\end{itemize}

The rest of this paper is organized as follows.
In Section~\ref{Princ}, we first illustrate
the implementation of the DRIS and formulate the ACA channel induced by the DRIS, where the time-varying reflection coefficients of the DRIS are generated by an independent DRIS controller without any coordination with the bistatic ISAC system. In Section~\ref{DIRSFPJComm}, we present the bistatic sensing model and the communication model in the presence of a DRIS, where channel models of all wireless channels involved are built.
In Section~\ref{ISACDesign}, we first formulate the ISAC waveform design problem in the presence of a DRIS 
and then introduce several trade-off factors to balance sensing and communication performance.
A theoretical lower bound of SINRs is derived to quantify the DRIS-induced ACA interference 
on the communication performance of the bistatic ISAC system.
Moreover,  we develop the CRLB to characterize 
the sensing performance in the presence of a DRIS.
Simulation results are presented in Section~\ref{ResDis} to demonstrate the impact of DRIS on the bistatic ISAC system and to validate the theoretical derivations.
Finally, conclusions are given in Section~\ref{Conclu}.

\emph{Notation:} We employ bold capital letters for a matrix, e.g., ${\bf{X}}$, lowercase bold letters for a vector, e.g., ${\boldsymbol x}_l$, 
and italic letters for a scalar, e.g., $k$. 
The superscripts $(\cdot)^{T}$ and $(\cdot)^{H}$ represent the transpose and the Hermitian transpose,  
and the symbols $\|\cdot\|_{\rm F}$ and $|\cdot|$  represent the Frobenius norm and the absolute value. %Moreover,where ${\mathop{\rm Re}\nolimits} \left({\cdot}\right)$ indicates the real part.
Moreover, $\mathbb{E}\!\left[\cdot\right]$ stands for the expectation.

\section{System Description}\label{Princ}
In this section,  we first illustrate
the implementation and design constraints of the DRIS,
and  formulate the ACA channel induced by the DRIS in Section~\ref{DISCOJamm}.
The reflection coefficients of the DRIS are random and time-varying, 
which are generated by an independent DRIS controller and are independent of the bistatic ISAC system.
In Section~\ref{DIRSFPJComm}, we present the bistatic sensing model and the communication model in the presence of a DRIS, 
where all involved wireless channels are characterized.

\subsection{Disco RIS-Based Active Channel Aging}\label{DISCOJamm}
\begin{figure}[!t]
    \centering
    \includegraphics[scale=0.62]{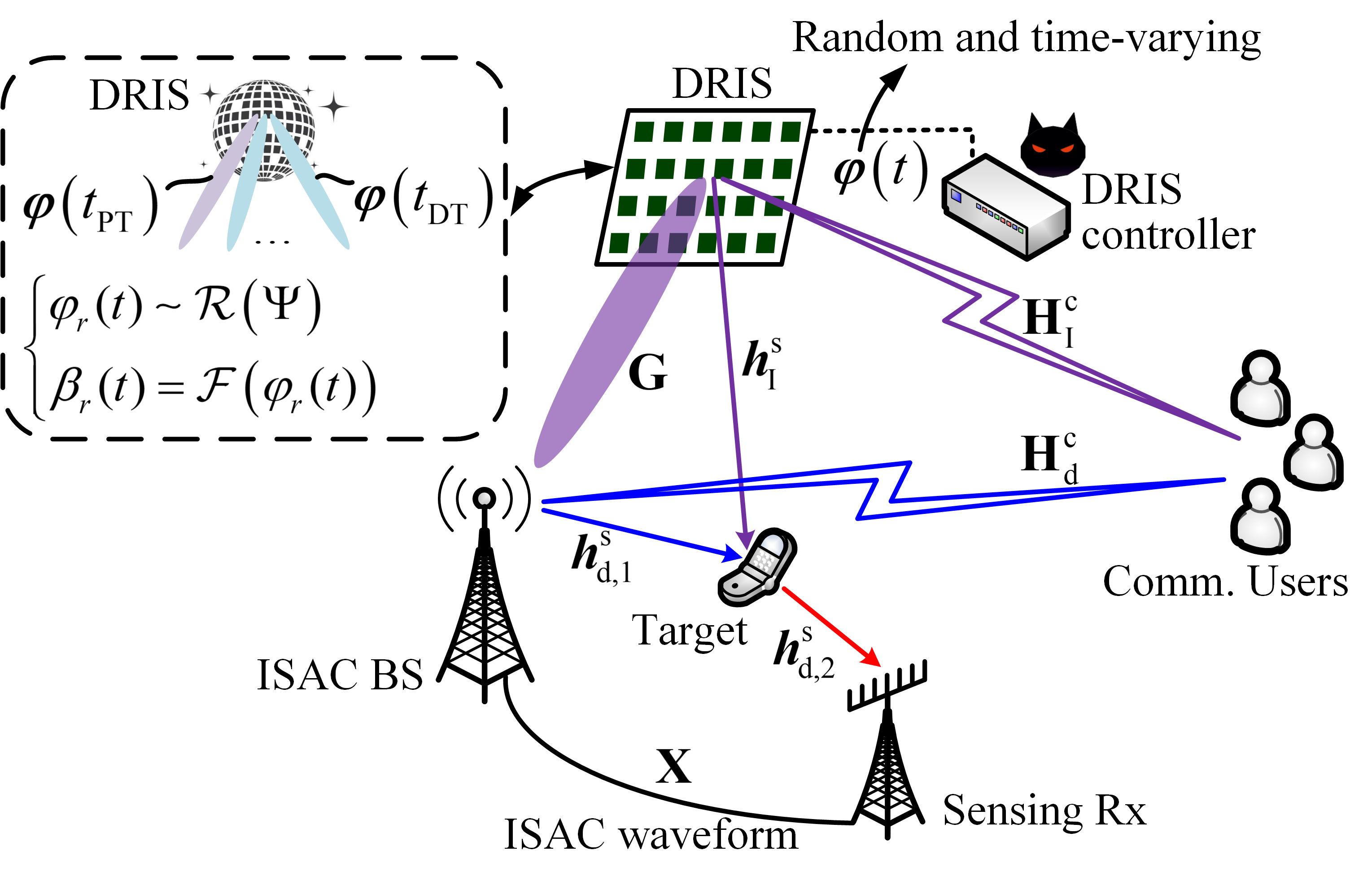}
    \caption{Illustration of a bistatic ISAC system with a bistatic radar configuration
    in the presence of a DRIS,
    where the time-varying DRIS reflection coefficients are randomly generated by the DRIS controller.}
    \label{fig1}
\end{figure}

Fig.~\ref{fig1} schematically illustrates a bistatic ISAC system 
based on bistatic sensing~\cite{MUIImpact,BistaticRadar} 
in the presence of a DRIS,
where the ISAC BS equipped with an $N_{\rm B}$-element uniform linear array (ULA) 
communicates with $K_{\rm c}$ single-antenna communication users 
while simultaneously detecting a target in a collaborative bistatic radar setting.
Furthermore, the sensing receiver is equipped with an $N_{\rm S}$-element ULA.
In this work, 
we adopt a robust assumption that
the DRIS-based FPJ has no information about the communication users, such as their location knowledge.
Therefore, we assume that the DRIS is deployed near the ISAC BS to maximize its impact~\cite{DIRSTWC,TWCAnti}.
Furthermore, the DRIS with $N_{\rm D} = N_{{\rm D},h} \times N_{{\rm D},v}$ reflective elements is implemented using PIN diodes, whose ON/OFF behavior permits only discrete phase shifts and amplitude values.
More specifically, the discrete reflective phase shifts are chosen from a set 
$\Psi = \left\{ {\phi_{1}},\ldots,  {\phi_{2^{b}}} \right\} \subset  [-\pi,\pi]$
 with $b$-bit quantization phase shifts,
and the corresponding time-varying amplitude of the $r$-th DRIS reflective element  $\beta_r$ 
is a function of ${\varphi}_{ r}(t)$~\cite{MyDISCOISACWCL},
i.e., $\beta_r(t) = {\mathcal F}\left({\varphi}_{ r}(t)\right) \in {\Xi} = \left\{\mu_1, \ldots, \mu_{2^b}\right\} \subset  [0,1]$.
The time-varying phase shift of the $r$-th DRIS reflective element ${\varphi}_{ r}(t)$ ($r=1,\ldots, N_{\rm \!D})$  is randomly selected from $\Psi$,
and is assumed to follow a stochastic distribution termed as ${\varphi}_{ r}(t) \sim {\mathcal R}\!\left(\! \Psi\right)$.
As a result, the time-varying DRIS reflective vector ${\boldsymbol{\varphi} ({t})}$ displayed in Fig.~\ref{fig1}
is expressed by $\boldsymbol{\varphi}\left( {{t}} \right) = 
\left[ {{\beta _1}({t}){e^{j{\varphi _1}({t})}}, \ldots ,{\beta _{{N_{\rm{D}}}}}\!({t}){e^{j{\varphi _{{N_{\rm{D}}}}}({t})}}} \right]$.
In this work, we further impose a constraint on the stochastic distribution ${\mathcal R}$, 
i.e., $\mathbb{E}\!\left[{\beta _r}({t}){e^{j{\varphi _r}({t})}}\right] = 0$.
% Mathematically, we impose the following constraint on the DRIS:
% \begin{equation}
% \sum\limits_{i = 1}^{{2^b}} p_i \mu_i  e^{j{\phi_{i}}} = 0,
%     \label{DRISCon}
% \end{equation}
% where the probability that the $r$-th DRIS element takes ${\phi_{i}} \in \Psi$ is denoted as $p_i$, 
%  $i=1,\ldots, 2^b$, and $\sum\nolimits_{i = 1}^{{2^b}} p_i=1$.

In conventional TDD systems, the wireless channel is commonly assumed to be quasi-static over the channel coherence interval of duration $T_{\rm C}$.
Consequently, the CSI estimated during the pilot transmission (PT) phase
($0 < t_{\rm{P\!T}} \le T_{\rm{P}} \le T_{\rm{C}}$) can be exploited to design the waveform for the subsequent data transmission (DT) phase ($T_{\rm{P}} < t_{\rm{D\!T}} \le T_{\rm{C}}$).
Specifically, the ISAC BS first learns the CSI during the PT phase via methods 
such as the least square (LS) algorithm.
However, in the bistatic ISAC system in the presence of a DRIS,
the DRIS randomly changes its reflection coefficients
with a period about the same as the length of the PT phase~\cite{MyWCMag,TWCAnti}.
As a result, the ACA induced by the DRIS breaks the channel reciprocity.
Mathematically, the CSI$\footnote{We assume that perfect CSI is available as imperfect CSI is not the primary focus of our ISAC scenario, and its impact has been extensively studied~\cite{RefSLNR,RefSLNRadd}. Under this assumption, we can focus directly on the analytical results in the presence of a DRIS, which, to the best of our knowledge, has not been reported before.}$ 
of communication users estimated during the PT phase is written as
\begin{equation}
    {{\bf{H}}_{{\rm{\!P\!T}}}^{\rm c}} =  {\bf{H}}_{\rm{d}}^{\rm c} \!+\! {\bf{H}}_{\rm{D}}^{\rm c}\!\left(t_{\rm{P\!T}}\right) \\
     = {{\bf{H}}_{\rm{d}}^{\rm c} \!+\! {{\bf{G}} }{ \rm{diag}\!} \left( {\boldsymbol{\varphi} ({t_{\rm{P\!T}}})} \right){\! \bf{H}}_{\rm{I}}^{\rm c} },
    \label{PTCSI}
\end{equation}
where ${\bf{H}}_{\rm{d}}^{\rm c}$ 
represents the direct channel between the ISAC BS and the communication users,
%that is unchanged during the channel coherence time,
and ${\bf{H}}_{\rm{D}}^{\rm c}\!\left(t_{\rm{P\!T}}\right)$ denotes 
the DRIS-jammed communication channel between the ISAC BS and the communication users during the PT phase.
In~\eqref{PTCSI}, $\bf G$ and ${\bf H}_{\rm I}^{\rm c}$ denote the channels between the ISAC BS and the DRIS
and between the DRIS and communication users, respectively.

On the one hand, the DRIS is deployed close to the ISAC BS to 
make its impact as large as possible.
On the other hand, the DRIS should be embedded in massive reflective elements to cope with the multiplicative
propagation path loss of the DRIS-jammed communication channels~\cite{AORIS}.
Consequently, we model the channel $\bf G$ between the ISAC BS and the DRIS based on the near-field model~\cite{NearfieldMoRef,NearfieldMo2},
given by
\begin{equation}
    {{\bf{G}}} \!=\!\! {\sqrt{{\mathscr{L}}_{\rm G}}} {\widehat{\bf{G}}} 
    \!=\! \!{\sqrt{{\mathscr{L}}_{\rm G}}}\!\!\left(\!\!{\sqrt {\frac{\varepsilon }{{1\! +\! \varepsilon }}}{\widehat{\bf{G}}}^{{\rm{LOS}}}} 
    \!+\!
    {\sqrt {\frac{1}{{1 \!+\! \varepsilon }}}{\widehat{\bf{G}}}^{{\rm{NLOS}}}}\!\right),
\label{Ricianchan}
\end{equation}
where ${\mathscr{L}}_{\rm G}$ stands for the large-scale channel fading coefficient of $\bf G$, 
and $\varepsilon$ represents
the Rician factor of $\bf G$, i.e., the ratio of the signal power
in the line-of-sight (LOS) component to the scattered power
in the non-line-of-sight (NLOS) component.

In~\eqref{Ricianchan}, the NLOS component ${{\widehat {\bf{G}}}^{{\rm{NLOS}}}}$ is assumed to follow Rayleigh fading, 
with elements that satisfy $\left[{{\widehat {\bf{G}}}^{{\rm{NLOS}}}}\right]_{n,r} \sim \mathcal{CN}\left(0,1\right), n=1,\ldots, N_{\rm B}$ and $r=1,\ldots, N_{\rm D}$.
Meanwhile,
the LOS component ${{\widehat {\bf{G}}}^{{\rm{LOS}}}}$ is
modeled as~\cite{TWCAnti,NearfieldMoRef,NearfieldMo2,NearfieldMo1}
\begin{equation}
\left[{\widehat {\bf{G}}}^{{\rm{LOS}}}\right]_{n,r} = {e^{ - j\frac{{2\pi }}{\lambda }\left( {{D_n^r} - {D_n}} \right)}} ,
\label{GLOS}
\end{equation}
where $\lambda$ denotes the wavelength of the ISAC waveform $\bf X$, and ${D_n^r}$ and ${D_n}$ represent the distance between the $n$-th antenna and the $r$-th DRIS reflective element,
and the distance between the $n$-th antenna of the ISAC BS's ULA and the centre (origin) of the DRIS, respectively. 
Both the spacing between adjacent DRIS reflecting elements and the spacing between adjacent transmit antennas are set to $d ={\lambda}/2$.
Furthermore, the locations of the first transmit antenna and the first DRIS element are taken as the reference locations of the ISAC BS and the DRIS, respectively.

Moreover, ${\bf H}^{\rm c}_{\rm d}$ and ${\bf H}^{\rm c}_{\rm I}$ are constructed based on the far-field model. 
In other words, ${\bf H}^{\rm c}_{\rm d}$ and ${\bf H}^{\rm c}_{\rm I}$ are assumed to follow Rayleigh fading~\cite{BookFarFeild}. Specifically, we have
\begin{alignat}{1}
&{{\bf{H}}^{\rm c}_{\rm{I}}} = {\left({{\bf{L}}^{\rm c}_{\rm I}}\right)^{\frac{1}{2} }}  {\widehat {\bf{H}}^{\rm c}_{\rm{I}}}
= \left[ {{\sqrt{{{\mathscr{L}}^{\rm c}_{{\rm I},1}}}}{{\widehat {\boldsymbol{h}}}^{\rm c}_{{\rm I},1}},  \ldots ,{\sqrt{{{\mathscr{L}}^{\rm c}_{{\rm I},K}}}}{{\widehat {\boldsymbol{h}}}^{\rm c}_{{\rm I},K_{\rm c}}}} \right], \label{HIkeq}\\
&{{\bf{H}}^{\rm c}_{\rm{d}}} =\left({{\bf{L}}_{\rm d}^{\rm c}}\right)^{\frac{1}{2}} {\widehat {\bf{H}}^{\rm c}_{\rm{d}}} = \left[ {{\sqrt{{{\mathscr{L}}^{\rm c}_{{\rm d},1}}}}{{\widehat {\boldsymbol{h}}}^{\rm c}_{{\rm d},1}}, \ldots ,{\sqrt{{{\mathscr{L}}^{\rm c}_{{\rm d},K}}}}{{\widehat {\boldsymbol{h}}}^{\rm c}_{{\rm d},K_{\rm c}}}} \right],
\label{Hdkeq}
\end{alignat}
where the elements of the $K_{\rm c}\times K_{\rm c}$ diagonal matrices ${\bf{L}}^{\rm c}_{\rm I} = {\rm{diag}}\left({{\mathscr{L}}^{\rm c}_{{\rm I},1}}, \ldots,{{\mathscr{L}}^{\rm c}_{{\rm I},K_{\rm c}}}\right)$ 
and ${\bf{L}}^{\rm c}_{\rm d} = {\rm{diag}}\!\left(\!{{\mathscr{L}}^{\rm c}_{{\rm d},1}}, 
\ldots,{{\mathscr{L}}^{\rm c}_{{\rm d},K_{\rm c}}}\!\right)$ denote the large-scale channel fading coefficients, 
which are assumed to be independent~\cite{LDetector1}. 
Furthermore, the elements of ${\widehat {\bf{H}}^{\rm c}_{\rm{I}}}$ 
and ${\widehat {\bf{H}}^{\rm c}_{\rm{d}}}$ 
are assumed to be i.i.d. Gaussian random variables~\cite{BookFarFeild} 
and defined as $\left[{\widehat {\bf{H}}_{\rm{I}}}\right]_{r,k},\left[{\widehat {\bf{H}}_{\rm{d}}}\right]_{n,k} \sim \mathcal{CN}\left(0,1\right)$, where $r=1,\ldots, N_{\rm D}$, $n=1,\ldots, N_{\rm B}$, and $k =1,\ldots,K_{\rm c}$.

During the subsequent DT phase, the wireless channels of the communication users ${{\bf{H}}_{{\rm{\!D\!T}}}^{\rm c}}$ are written as
\begin{equation}
    {{\bf{H}}_{{\rm{\!D\!T}}}^{\rm c}} =  {\bf{H}}_{\rm{d}}^{\rm c} \!+\! {\bf{H}}_{\rm{D}}^{\rm c}\!\left(t_{\rm{D\!T}}\right)
     = {{\bf{H}}_{\rm{d}}^{\rm c} \!+\! {{\bf{G}} }{ \rm{diag}\!} \left( {\boldsymbol{\varphi} ({t_{\rm{D\!T}}})} \right){\! \bf{H}}_{\rm{I}}^{\rm c} },
    \label{DTCSI}
\end{equation}
where ${\bf{H}}_{\rm{D}}^{\rm c}\! (t_{\rm{D\!T}} )$ stands for the DRIS-jammed  communication channel 
between the ISAC BS and the communication users during the DT phase.

Note that the DRIS randomly changes its reflection coefficients,
and thus ${\bf{H}}_{\rm{D}}^{\rm c}\! (t_{\rm{P\!T}} )$ in~\eqref{PTCSI} and ${\bf{H}}_{\rm{D}}^{\rm c}\! (t_{\rm{D\!T}} )$ in~\eqref{DTCSI}
are different.
The update interval of the DRIS reflection coefficients is typically assumed to be on the same order as the duration of the PT phase~\cite{MyWCMag,TWCAnti}. 
Consequently, the ACA induced by the DRIS leads to much faster channel aging, reducing the effective channel coherence time to the same order as the PT duration, i.e., ${T_{\rm{C}}} \to {T_{\rm{P}}}$.
Specifically, we define the DRIS-based ACA channel as
\begin{equation}%{alignat}{1}
    {{\bf{H}}^{\rm c}_{{\rm{\!A\!C\!A}}}}  \!=\! {{\bf{H}}^{\rm c}_{{\rm{\!D\!T}}}}\!-\!{{\bf{H}}^{\rm c}_{{\rm{\!P\!T}}}}  
    \!=\!{{\bf{G}}}{ \rm{diag}\!} \left( {\boldsymbol{\varphi} ({t_{\rm{\!D\!T}}})} \!-\!{\boldsymbol{\varphi} ({t_{\rm{\!P\!T}}})} \right){\! \bf{H}}_{\rm{I}}^{\rm c}.
    \label{ACAEq}
\end{equation}

\subsection{ISAC Model in the Presence of a DRIS}\label{DIRSFPJComm}
\underline{\textit{Communication Model:}}
During the DT phase of the channel coherence time, 
the ISAC BS transmits $L$ constellation symbol vectors to the users, 
which is written as ${\bf S} = [{\boldsymbol s}_{ 1},\ldots, {\boldsymbol s}_{ L}] \in {\mathbb{C} }^{K_{\rm c} \times L}$.
Then, the corresponding $L$ received symbols at the users are denoted by ${\bf{Y}_{\rm c}} = [{\boldsymbol y}^{\rm c}_{1},\ldots, {\boldsymbol y}^{\rm c}_{L}] \in {\mathbb{C} }^{K_{\rm c} \times L}$.
More specifically, the $l$-th received symbol vector is expressed as
\begin{equation}%{alignat}{1}
    {\boldsymbol y}^{\rm c}_{l} = {\boldsymbol s}_{l} + \underbrace {\left( {{{\bf{H}}^{\rm c}_{{\rm{PT}}}}{{\boldsymbol x}_{ l}} - 
    {\boldsymbol s}_{ l}} \right)}_{\rm{MU\; interference}} + 
    \underbrace {{{\bf{H}}^{\rm c}_{{\rm{\!A\!C\!A}}}} {{\boldsymbol x}_{ l}}}_{\rm{ACA\; interference}} + {{\boldsymbol n}^{\rm c}_{ l}},
    \label{SigRx}
\end{equation}
where  ${{\boldsymbol x}_{ l}} \in {\mathbb{C} }^{N_{\rm B} \times 1}$ is 
the $l$-th column of an $\left(N_{\rm B} \times L\right)$-dimensional transmitted signal matrix
${\bf X}$ that is used as the ISAC waveform
for sensing and communication functions~\cite{FLiuISAC,FLiuISAC1},
and $ {{\boldsymbol n}^{\rm c}_{ l}}  \in \mathbb{C}^{K_{\rm c} \times 1}$ is AWGN 
with zero mean and variance $\sigma_{\rm c}^2$, 
i.e., ${{n}^{\rm c}_{k,l}} \sim \mathcal{CN}\left(0,\sigma_{\rm c}^2\right)$ ($k=1,\ldots,K_{\rm c}, l = 1,\ldots,L$).

% In an ISAC system, the ISAC waveform $\bf X$ is designed based on the CSI estimated during the PT phase~\cite{FLiuISAC}.
% Defining that ${{\bf{C}}_{\bf{X}}} = \frac{{\bf{X}}{{\bf{X}}^H}}{L}$, 
% ${\bf C}_{\bf X}$ needs to satisfy the following constraint to achieve the best sensing performance, i.e., 
% \begin{equation}%{alignat}{1}
%     {{\bf{C}}_{\bf{X}}} =  \frac{{{P_0} }}{{{N_{\rm{B}}}}}{{\bf{I}}_{{N_{\rm{B}}}}},
%     \label{ISACWaveLim}
% \end{equation} 
% where $P_0$ denotes the total transmit power of the data in each DT phase 
% and ${\bf I}_{N\!_{\rm B}}$ is an $N_{\rm B} \times N_{\rm B}$ identity matrix. 

From~\eqref{SigRx}, we can see that the DRIS imposes ACA interference
on the received symbols in addition to multi-user (MU)  interference.
Referring to~\cite{DIRSTWC,FLiuISAC,EquWir}, the signal-to-interference-plus-noise ratio (SINR)
at the $k_{\rm}$-th communication user is represented by
\begin{equation}%{alignat}{1}
    {\gamma _k} \!=\! \frac{{\mathbb{E}\!\left[ \!{{{\left| {{s_{k,l}}} \right|}^2}} \right]}}{{{\mathbb E}\!\left[\! \left|\!{\left(\! {{{ ( {\boldsymbol{h}_{{\rm{PT}},k}^{\rm{c}}} )}^H}{\boldsymbol{x}_l} \!-\! {s_{k,l}}} \right) 
    \!+\! {{ \left( \!{\boldsymbol{h}_{{\rm{A\!C\!A}},k}^{\rm{c}}} \right)}^H}{\boldsymbol{x}_l}} \right|^2\right] \!+\! {\sigma_{\rm c}^2}}}.
    \label{SINRDRIS}
\end{equation}
Consequently, the sum rate can be computed based on~\eqref{SINRDRIS}, i.e., $R_{\rm{sum}} = \sum\nolimits_{k = 1}^{K_{\rm c}} {{R_k}}  = \sum\nolimits_{k = 1}^{K_{\rm c}} {{{\log }_2}\left( {1 + {\gamma _k}} \right)}$.

It is worth noting that the strict equality constraint in~\eqref{ISACWaveLim} ensures that the ISAC waveform $\bf X$ 
possesses the same properties as the best sensing waveform.
However, this ISAC waveform $\bf X$ suffers from performance degradation
because the MU interference term, i.e., 
$({\bf H}^{\rm c}_{\rm{P\!T}}{{\boldsymbol x}_{ l}} - {{\boldsymbol s}_{ l}})$, 
is not taken into account.
Based on~\eqref{SINRDRIS}, ignoring the MU interference 
when designing $\bf X$ leads to a reduction in the communication performance metric, 
i.e., $R_{\rm{sum}}$.
Moreover, the DRIS excites an additional ACA interference term
${{\bf{H}}^{\rm c}_{{\rm{\!A\!C\!A}}}} {{\boldsymbol x}_{ l}}$,
which further degrades the communication performance.

\underline{\textit{Sensing Model:}}
To quantify the sensing performance,
we evaluate the estimation accuracy of the target angles~\cite{NearfieldMo1,RadarMer,RadarMer1}.
The $L$ symbol vectors are reflected by the target and then captured at the sensing receiver. %and denoted by ${{\bf{Y}}_{\rm{ s}}} = [{\boldsymbol y}^{\rm s}_{1},\ldots, {\boldsymbol y}^{\rm s}_{L}] \in {\mathbb{C} }^{N_{\rm S} \times L}$
Mathematically, the $l$-th received sensing signal vector is expressed as 
\begin{equation}%{alignat}{1}
 {{\boldsymbol{y}}_l^{\rm{ s}}} =  \chi\!\left({\boldsymbol{h}_{{\rm d},2}^{\rm s}}\right)^{\!H} \!\!\left( {\boldsymbol{h}_{{\rm d},1}^{\rm s}} 
 + {\boldsymbol{h}_{\rm D}^{\rm s}}(t_l)\right) \!
    {{\boldsymbol x}_{ l}} + {{\boldsymbol n}^{\rm s}_{ l}} ,
    \label{SensingChannel}
\end{equation}
where $l=1, \ldots, L$, and $0\le \!\chi \!\le 1$ denotes the reflection cross-section coefficient of the target. 
Here,
${\boldsymbol{h}_{{\rm d},1}^{\rm s}} \in \mathbb{C}^{1\times N_{\rm B}}$ represents the direct sensing path 
between the ISAC BS and the target, %=  [{{h}_{1,1}^{\rm s}}, \ldots,{{h}_{1,N_{\rm B}}^{\rm s}} ]^T
${\boldsymbol{h}}_{\rm{D}}^{\rm{s}}\!(t)\in \mathbb{C}^{1\times N_{\rm B}}$ %= [{{h}_{{\rm D},1}^{\rm s}}\!(t), \ldots,{{h}_{{\rm D},N_{\rm B}}^{\rm s}}\!(t) ]^T 
represents the time-varying DRIS-based sensing path,
${\boldsymbol{h}_{{\rm d},2}^{\rm s}}\in \mathbb{C}^{1\times N_{\rm S}}$ denotes the direct sensing path between 
the target and the sensing receiver, % [{{h}_{2,1}^{\rm s}}, \ldots,{{h}_{2,N_{\rm B}}^{\rm s}} ]^T
and ${{\boldsymbol n}^{\rm s}_{ l}}$ is AWGN with zero mean and variance $\sigma_{\rm s}^2$. 

% We assume that the direct sensing path between the ISAC BS and the target and that between the target and 
% the  sensing receiver are LoS. 
More specifically, ${\boldsymbol{h}_{{\rm d},1}^{\rm s}}$ and ${\boldsymbol{h}_{{\rm d},2}^{\rm s}}$ 
in~\eqref{SensingChannel} are modeled based on the ULA array response, 
which are given by~\cite{RadarMer1}
\begin{alignat}{1}
    {\boldsymbol{h}_{\rm{d},1}^{\rm s}}=\left[{h}_{1,1}^{\rm s},\ldots,{h}_{1,N\!_{\rm B}}^{\rm s}\right] 
    &= {\sqrt{{{\mathscr{L}}^{\rm s}_{{\rm {d}},1}}}}{ \boldsymbol{\alpha}}_{{\!N_{\rm B}}} \!\!\left({\theta_1}\right),\label{hds}\\
    {\boldsymbol{h}_{\rm{d},2}^{\rm s}}=\left[{h}_{2,1}^{\rm s},\ldots,{h}_{2,N\!_{\rm S}}^{\rm s}\right] 
    &= {\sqrt{{{\mathscr{L}}^{\rm s}_{{\rm {d}},2}}}}{ \boldsymbol{\alpha}}_{{\!N_{\rm S}}} \!\!\left({\theta_2}\right),\label{hd2s} 
\end{alignat}
where ${{\mathscr{L}} _{{\rm {d}},1}^{\rm s}}$ and ${{\mathscr{L}} _{{\rm {d}},2}^{\rm s}}$
 are the large-scale channel fading coefficients of ${\boldsymbol{h}_{\rm{d},1}^{\rm s}}$ and 
 ${\boldsymbol{h}_{\rm{d},2}^{\rm s}}$, respectively.
Furthermore, the steering vectors of the ISAC BS and the sensing receiver are respectively given by
\begin{alignat}{1}
    { \boldsymbol{\alpha}}_{{\!N_{\rm B}}} \!\!\left({\theta_1}\right) & =\!  
    {\left[ \!{1,{e^{j2\pi \Delta \!\sin \theta_1}}, \ldots ,{e^{j2\pi (N\!_{\rm B} \!-\! 1)\Delta \!\sin \theta_1}}} \right] },
    \label{LPA}\\
    { \boldsymbol{\alpha}}_{{\!N_{\rm S}}} \!\!\left({\theta_2}\right) & =\!  
    {\left[ \!{1,{e^{j2\pi \Delta \!\sin \theta_2}}, \ldots ,{e^{j2\pi (N\!_{\rm S} \!-\! 1)\Delta \!\sin \theta_2}}} \right] },
    \label{LPA2}
\end{alignat}
where $\theta_1$ and $\theta_2$ stand for the AoD from the ISAC BS
and the AoA at the sensing receiver, respectively.
Furthermore, $\Delta = \frac{d}{\lambda}$ represents the normalized antenna spacing of the ULAs,
where ${\lambda}$ is the wavelength
of the transmit ISAC waveform $\bf X$.
 
In~\eqref{SensingChannel}, ${\boldsymbol{h}}_{\rm{D}}^{\rm{s}}(t)$ 
can be further modeled as
\begin{equation}
    {\boldsymbol{h}}_{\rm{D}}^{\rm{s}}(t)  \!=\! {\bf G}{\rm diag}\!\left( {\boldsymbol{\varphi}}(t) \right) {{\boldsymbol{h}}_{\rm I}^{\rm s}} 
    \!=\!\! {\sqrt{{{\mathscr{L}}_{{\rm{G}}} }}}{\widehat {\bf{G}}}{\rm diag}\!\left(\! {\boldsymbol{\varphi}}(t)\right)\!{\sqrt{{{\mathscr{L}}_{{\rm{I}}}^{\rm s}}}} {\widehat{{\boldsymbol{h}}_{\rm I}^{\rm s}} } ,
    \label{hDs}
\end{equation}
where ${{\mathscr{L}}_{{\rm{I}}}^{\rm s}}$ is the large-scale channel fading coefficient of the path between the DRIS and the target.
For ease of presentation, we denote the cascaded large-scale channel fading coefficient of ${\boldsymbol{h}}_{\rm{D}}^{\rm{s}}\!(t)$
as ${{\mathscr{L}}_{{\rm{cas}}}^{\rm s}}={{\mathscr{L}}_{{\rm{G}}} }{{\mathscr{L}}_{{\rm{I}}}^{\rm s}}$.
Furthermore,  ${\widehat{{\boldsymbol{h}}_{\rm I}^{\rm s}} }$ in~\eqref{hDs}
can be built based on 
the uniform plane array (UPA) response, i.e., 
\begin{equation}%{alignat}{1}
    {\widehat{{\boldsymbol{h}}_{\rm I}^{\rm s}} }=
    {\boldsymbol{\alpha}}_{{\!N_{\rm D},h}} \!\left({\vartheta_h }\right)\otimes  {\boldsymbol{\alpha}}_{{\!N_{\rm D},v}} \!\left({\vartheta_v }\right),
    \label{hIs}
\end{equation}
where
${\vartheta_h }$ and ${\vartheta_v}$ stand for the horizontal and vertical
AoDs between the two-dimensional DRIS and the target, respectively, and 
$\otimes$  denotes the Kronecker product.

%Based on~\eqref{SensingChannel}, the sensing receiver estimates the AoD $\theta_1$ and the AoA $\theta_2$.
Without loss of generality, we focus on the estimation of the target's AoD and AoA.
As shown in~\eqref{SensingChannel}, however,
the time-varying DRIS-induced ACA term 
$\chi\!\left({\boldsymbol{h}_{{\rm d},2}^{\rm s}}\right)\!^{H}\!{\boldsymbol{h}_{\rm D}^{\rm s}}\!(t_{l}) {{\boldsymbol x}_l}$ 
is introduced in the presence of a DRIS, 
which influences the accuracy of the estimation of $\theta_1$ and $\theta_2$ from the 
2D-MUSIC algorithm compared to the case without a DRIS.

\section{Integrated Sensing and Communication in the Presence of a DRIS}\label{ISACDesign}
In Section~\ref{ProblemFor}, we first formulate the ISAC waveform design problem 
for the bistatic ISAC system in the presence of a DRIS.
Then, we propose an ISAC waveform optimization scheme for the bistatic ISAC system in the presence of the DRIS.
In Section~\ref{DRISJamm}, a theoretical lower bound of SINRs is derived to quantify the 
DRIS-induced ACA interference on the communication performance of the bistatic ISAC system.
Moreover, in Section~\ref{DRISJammSnesing}, 
we develop the CRLB to characterize 
the sensing performance in the presence of a DRIS.

\subsection{Problem Formulation and ISAC Waveform Design}~\label{ProblemFor}
According to~\eqref{SigRx}, the optimal ISAC waveform should jointly suppress both the MU interference and the ACA interference. 
However, due to the time-varying and random DRIS reflective vector ${\boldsymbol{\varphi} ({t})}$, the ISAC BS cannot estimate ${\bf H}^{\rm c}_{\rm{A\!C\!A}}$ effectively, which makes the aforementioned optimal design intractable in practice.
Therefore, the ISAC BS can only design the ISAC waveform by minimizing the MU interference. 
Mathematically, the ISAC waveform design problem is formulated as
\begin{alignat}{1}
    \left( {{\rm{P}}1} \right): &\mathop {\min}\limits_{{\bf{X}}} \left\| {{\bf{H}}_{{\rm{PT}}}^{\rm{c}}{\bf{X}} - {\bf{S}}} \right\|_{\rm{F}}^2 \label{P1forS}\\
    &{\rm{s}}.{\rm{t}}.\;\; {{\bf{C}}_{\bf{X}}} =  \frac{{{P_0} }}{{{N_{\rm{B}}}}}{{\bf{I}}_{{N_{\rm{B}}}}}
    \label{ISACWaveLim},
\end{alignat}
where $P_0$ denotes the total transmit power at the ISAC BS
  and ${\bf I}_{N\!_{\rm B}}$ is an $N_{\rm B} \times N_{\rm B}$ identity matrix. 

In $\left( {{\rm{P}}1} \right)$, the strict equality constraint in~\eqref{ISACWaveLim} ensures that the ISAC waveform possesses the same properties as the optimal sensing waveform. 
However, this constraint significantly degrades the communication performance. 
To balance the trade-off between the sensing performance and the communication rate,
we introduce a trade-off factor $\kappa$ ($0 \le \kappa \le 1$) into $\left( {{\rm{P}}1} \right)$.
Denoting the solution to $\left( {{\rm{P}}1} \right)$ as ${\bf X}_0$, the following Pareto optimization problem can be obtained:
\begin{alignat}{1}
    \left( {{\rm{P}}2} \right): &\mathop {\min}\limits_{{\bf{X}}} \kappa \left\| {{\bf{H}}_{{\rm{PT}}}^{\rm{c}}{\bf{X}} - {\bf{S}}} \right\|_{\rm{F}}^2 + (1-\kappa)\left\| {{\bf{X}} - {{\bf{X}}_0}} \right\|_{\rm{F}}^2 \\
    &\,{\rm{s}}.{\rm{t}}.\; {\rm{tr}}\!\left({\bf{X}}{\bf{X}}^H\!\right) = P_0L .
    \label{P1forISA}
\end{alignat}
For different $\kappa$, the solution to $\left( {{\rm{P}}2} \right)$ 
yields different trade-offs between sensing and communication performance.
In particular, a smaller $\kappa$
prioritizes sensing performance at the cost of communication performance.

To solve $\left( {{\rm{P}}2} \right)$, we first obtain ${\bf X}_0$ by solving $\left( {{\rm{P}}1} \right)$. 
Although $\left( {{\rm{P}}1} \right)$ is a non-convex problem,  it is a classical orthogonal Procrustes problem~\cite{P1slove}.
Therefore, a closed-form solution to $\left( {{\rm{P}}1} \right)$  can be derived based on singular value decomposition (SVD), i.e.,
\begin{equation}%{alignat}{1}
    {{\bf{X}}_0} = \sqrt {\frac{{{P_0}L}}{N_{\rm B}}} {\bf{U}}{{\bf{I}}_{N_{\rm B} \times L}}{\!{\bf{V}}^H}
    = \left[{\boldsymbol{x}}_{0,1},\ldots, {\boldsymbol{x}}_{0,L}\right],
    \label{ForX0}
\end{equation}
where ${\bf U}\in {\mathbb{C}}^{N_{\rm B}\times N_{\rm B}}$ and ${\bf V}\in {\mathbb{C}}^{L \times L}$ 
are the left and right singular value matrices of $\left({\bf H}^{\rm c}_{\rm{P\!T}}\right)^H{\bf S}$, 
such that ${\bf U} {\boldsymbol \Sigma} {\bf V}^H= \left({\bf H}^{\rm c}_{\rm{P\!T}}\right)^H{\bf S}$
with a (rectangular) diagonal matrix ${{\boldsymbol \Sigma}}\in {\mathbb{R}}^{N_{\rm B}\times L}$ 
containing the singular values.

To facilitate the solution of $\left( {{\rm{P}}2} \right)$, we reformulate it into the following form 
using a small trick in~\cite{FLiuISAC},
\begin{alignat}{1}
    \left( {{\rm{P}}2}-{{\rm{E}}1} \right): &\mathop {\min}\limits_{{\bf{X}}} \left\| {{\bf{AX}} - {\bf{B}}} \right\|_{\rm{F}}^2 \\
    \nonumber
    &{\rm{s}}.{\rm{t}}.~\eqref{P1forISA},
\end{alignat}
where  ${\bf A} = {\!\left[\! {\sqrt \kappa  {{({\bf{H}}_{{\rm{PT}}}^{\rm{c}})}^T},\sqrt {1 \!-\! \kappa } {{\bf{I}}_N}} \!\right]\!^T}$ and ${\bf{B}} = {\left[\! {\sqrt \kappa  {{\bf{S}}^T},\sqrt {1 \!-\! \kappa } {\bf{X}}_0^T}\! \right]\!^T}$.

Although $\left( {{\rm{P}}2}-{{\rm{E}}1} \right)$ is a non-convex quadratically constrained quadratic programming (QCQP) problem, it can be equivalently reformulated as a semidefinite program (SDP) via semidefinite relaxation (SDR).
Moreover, $\left( {{\rm{P}}2}-{{\rm{E}}1} \right)$ involves only one quadratic constraint, i.e.,~\eqref{P1forISA}.
Therefore, the resulting rank-one solution is globally optimal~\cite{TightSDR}.
% Based on the solution $\bf X$ from $\left( {{\rm{P}}2}-{{\rm{E}}1} \right)$, 
% we can define the resulting covariance matrix
%  ${\widetilde {\bf{C}}_{\bf X}} = \frac{{\bf X}{\bf X}^H}{L}$,
% which in general does not satisfy the constraint~\eqref{ISACWaveLim}.
 
\subsection{Impact of DRIS on Communications}~\label{DRISJamm}
For a given ISAC waveform $\bf X$, the sum rate 
is affected not only by the MU interference resulting from the sensing functionality
 but also by the ACA interference induced by the DRIS.
The impact of MU interference has been extensively investigated and quantified in the existing ISAC literature,
e.g., \cite{FLiuISAC,FLiuISAC1}.
Therefore, in this section, we focus on characterizing the impact of the DRIS-induced ACA interference on 
the communication performance of the bistatic ISAC system.
To facilitate the analysis, 
we rewrite the aggregate interference term in the denominator of the SINR in~\eqref{SINRDRIS} as 
\begin{alignat}{1}%{alignat}{1}
    \nonumber
    &{\mathscr{I}} \! = \left| {{{ ( {\boldsymbol{h}_{{\rm{PT}},k}^{\rm{c}}} )}^H}{\boldsymbol{x}_l} \!-\! {s_{k,l}}} \right|^2 \!\!+\! 
    \left({{{ ( {\boldsymbol{h}_{{\rm{PT}},k}^{\rm{c}}} )}^H}{\boldsymbol{x}_l} \!-\! {s_{k,l}}} \right){\boldsymbol{x}^H_l}{\boldsymbol{h}_{{\rm{ACA}},k}^{\rm{c}}} \\
    &\; +  \!\!({{{\! (\! {\boldsymbol{h}_{{\rm{PT}},k}^{\rm{c}}} )}^H}{\!\boldsymbol{x}_l} \!-\! {s_{k,l}}}  )^H{{ \!( {\boldsymbol{h}_{{\rm{ACA}},k}^{\rm{c}}} )}\!^H}\!{\boldsymbol{x}_l} \!+ \!\!\left|\!{{ ( {\boldsymbol{h}_{{\rm{ACA}},k}^{\rm{c}}} )}\!^H}{\boldsymbol{x}_l} \right|^2.
    \label{RewInterf}
\end{alignat}

With the exception of the first term in~\eqref{RewInterf}, 
the remaining three terms are generated from the DRIS-induced ACA channel ${\bf{H}}^{\rm c}_{\rm{A\!C\!A}}$,
thereby capturing the ACA interference impact of the DRIS. 
To evaluate the expectation of these ACA-related 
interference terms, 
we first characterize the statistical distribution of ${\bf{H}}^{\rm c}_{\rm{A\!C\!A}}$.
% Note that the DRIS requires a large number of reflective elements to 
% overcome the multiplicative propagation loss in the cascaded DRIS-based channels. 
Specifically, the i.i.d. elements of ${\bf{H}}^{\rm c}_{\rm{A\!C\!A}}$ have 
the statistical characteristics given
in Proposition~\ref{Proposition1}.
\newtheorem{proposition}{Proposition}
\begin{proposition}
    \label{Proposition1}
    The elements of ${\bf{H}}^{\rm c}_{\rm {\!A\!C\!A}}$ converge in distribution to $\mathcal{CN}\!\left( {0,  {{{\mathscr{L}}\!_{{\rm {cas}},k}} {N\!_{\rm D}} {\overline \mu} } } \right)$ as $N_{\rm D} \to \infty$, i.e.,
    \begin{equation}
        {\left[ {{\bf H}^{\rm c}_{\rm {\!A\!C\!A}}} \right]_{n,k}} \mathop  \to \limits^{\rm{d}} \mathcal{CN}\!\!\left( {0,  {{{\mathscr{L}}^{\rm c}_{{\rm {cas}},k}} {N\!_{\rm D}}{\overline \mu} } } \right), \forall n,k,
        \label{HDSta}
    \end{equation}
where ${{\mathscr{L}}^{\rm c}_{{\rm {cas}},k}} = {{\mathscr{L}}_{{\rm {G}}}} 
{{\mathscr{L}}^{\rm c}_{{\rm {I}},k}}$ represents the cascaded large-scale channel
fading coefficient of the DRIS-based channel between the ISAC BS and the $k$-th communication user,
$\overline \mu = \sum\nolimits_{i1 = 1}^{{2^b}} \sum\nolimits_{i2 = 1}^{{2^b}}{p_{i1}}{p_{i2}} \big( {\mu _{i1}^2} + {\mu _{i2}^2} - 2{\mu _{i1}}{\mu _{i2}}$ $\cos ( {{\phi_{i1}} - {\phi_{i2}}} ) \big)$,
    ${\mu _{i1}}$, ${\mu _{i2}} \in \Xi$, ${\phi_{i1}}$, ${\phi_{i2}} \in \Psi$, and ${p_{i1}}, {p_{i2}}$ 
are the probabilities of taking the random phases ${\phi_{i1}}, {\phi_{i2}}$.
\end{proposition}

\begin{IEEEproof}
    See Appendix~\ref{AppendixA}.
\end{IEEEproof}

Based on Proposition~\ref{Proposition1}, 
the impact of the DRIS-induced ACA interference 
is  quantified by Theorem~\ref{Theorem1} below. 
\newtheorem{theorem}{Theorem}
\begin{theorem}
\label{Theorem1}
In the presence of DRIS-induced ACA interference, the lower bound on the SINR received at the $k$-th communication user ${\gamma _k}$ for an ISAC waveform obtained from $\left( {{\rm{P}}2}-{{\rm{E}}1} \right)$  converges in distribution to 
\begin{alignat}{1}%{alignat}{1}
   & {\gamma _k}  \mathop  \to \limits^{\rm{d}} \frac{{\mathbb{E}\!\!\left[\! {{{\left|{{s_{k,l}}} \right|}^2}} \right]}}{{{\mathbb E}\!\!\left[\! \left|\! {\left( {{{ ( {\boldsymbol{h}_{{\rm{PT}},k}^{\rm{c}}} )}^H}{\boldsymbol{x}_l} - {s_{k,l}}} \right) }\! 
     \right|^2\right] \!\!+\!{\mathbb E}\!\!\left[\left|\! {\boldsymbol{h}_{{\rm{A\!C\!A}},k}^{\rm{c}}} {\boldsymbol{x}_{l}}\!\right|^2 \!\right]\! +\!{\sigma^2}}}, \label{xxxxx}\\
    & \;\; \ge  \frac{{\mathbb{E}\!\!\left[ \!{{{\left|{{s_{k,l}}} \right|}^2}} \right]}}{{{\mathbb E}\!\!\left[ \!\left|\! {\left( {{{ ( {\boldsymbol{h}_{{\rm{PT}},k}^{\rm{c}}} )}^H}{\boldsymbol{x}_l} - {s_{k,l}}} \right) }\!  \right|^2\right] 
    \!+ \!{ P_0 {{{\mathscr{L}}^{\rm c}_{{\rm {cas}},k}} {N\!_{\rm D}}{\overline \mu} }} +{\sigma^2}}}.
    \label{SINRDRISMath}
\end{alignat}
\end{theorem}

\begin{IEEEproof}
The expectation of ${\mathscr{I}}$ in~\eqref{RewInterf} can be written as
\begin{alignat}{1}%{alignat}{1}
    \nonumber
    &{\mathbb{E}\!\left[{\mathscr{I}}\!\right]\!}\! =\! \! {\mathbb{E}\!\!\left[\!\left|\! {{{ ( {\boldsymbol{h}_{{\rm{PT}},k}^{\rm{c}}} )}^{\!H}}\!{\boldsymbol{x}_l} \!-\!\! {s_{k,l}}} \right|^2\!\right]\! }\!\!+\!\! 
    {\mathbb{E}\!\!\left[\! \left(\!{{{ \!( {\boldsymbol{h}_{{\rm{PT}},k}^{\rm{c}}} )}^{\!H}}\!{\boldsymbol{x}_l} \!-\!\! {s_{k,l}}} \!\right)\!{\boldsymbol{x}^H_l}\!{\boldsymbol{h}_{\!{\rm{A\!C\!A}},k}^{\rm{c}}}\!\right] } \\
    &+  \! \!{\mathbb{E}\!\!\left[\!\left({{{\! (\! {\boldsymbol{h}_{{\rm{PT}},k}^{\rm{c}}} )}^H}{\!\boldsymbol{x}_l} \!-\! {s_{k,l}}} \!\!\right)^{\!H}\!{{ \!( {\boldsymbol{h}_{{\rm{ACA}},k}^{\rm{c}}} )}^{\!H}}\!{\boldsymbol{x}_l}\right]} \!\!+ 
    \!\mathbb{E}\!\!\left[\!\left|{{ ( {\boldsymbol{h}_{{\rm{ACA}},k}^{\rm{c}}} )}^{\!H}}{\boldsymbol{x}_l} \right|^2\!\right].
    \label{InterferTermReduce}
\end{alignat}

Based on Proposition~\ref{Proposition1}, we have the following conclusion:
\begin{alignat}{1}%{alignat}{1}
    \nonumber
    &{\mathbb{E}\!\!\left[\!\left({{{ \!( {\boldsymbol{h}_{{\rm{PT}},k}^{\rm{c}}} \!)}^H}\!{\boldsymbol{x}_l} \!-\! {s_{k,l}}} \right){\boldsymbol{x}^H_l}{\boldsymbol{h}_{\!{\rm{A\!C\!A}},k}^{\rm{c}}}\!\right] }  \le\\
    &\;\;\;\;\;\;\;\;\;\;\;\;\;\;\;\;\;\; {\mathbb{E}\!\!\left[\!\left({{{ \!( {\boldsymbol{h}_{{\rm{PT}},k}^{\rm{c}}} \!)}^H}\!{\boldsymbol{x}_l} \!-\! {s_{k,l}}} \right){\boldsymbol{x}^H_l} \!\right] } 
    {\mathbb{E}\!\left[\!{\boldsymbol{h}_{\!{\rm{A\!C\!A}},k}^{\rm{c}}}\!\right] } 
    \mathop  \to \limits^{\rm{d}} 0.
    \label{InterferTermHACA}
\end{alignat}
Similarly, 
\begin{alignat}{1}%{alignat}{1}
    \nonumber
    &{\mathbb{E}\!\!\left[\! \left({{{\! (\! {\boldsymbol{h}_{{\rm{PT}},k}^{\rm{c}}} )}^H}{\!\boldsymbol{x}_l} \!-\! {s_{k,l}}} \right)^H{{ \!( {\boldsymbol{h}_{{\rm{ACA}},k}^{\rm{c}}} )}^H}\!{\boldsymbol{x}_l}  \! \right] }  \le\\
    &\;\;\;\;\;\;\;\;\;\;\;  {\mathbb{E}\!\!\left[\!  \left({{{\! (\! {\boldsymbol{h}_{{\rm{PT}},k}^{\rm{c}}} )}^H}{\!\boldsymbol{x}_l} \!-\! {s_{k,l}}} \right)^H{{ \!( {\boldsymbol{h}_{{\rm{ACA}},k}^{\rm{c}}} )}^H} \!\right] } 
    \!{\mathbb{E}\!\left[{\boldsymbol{x}_l}\right] } 
    \mathop  \to \limits^{\rm{d}} 0.
    \label{Interfer2TermHACA}
\end{alignat}

According to Proposition~\ref{Proposition1}, we have 
\begin{alignat}{1}%{alignat}{1}
    & {\mathbb E}\!\!\left[\!\left| \!{{ \left( \!{\boldsymbol{h}_{{\rm{A\!C\!A}},k}^{\rm{c}}}\!\right)\!}^H}\!{\boldsymbol{x}_l}\!\right|^2\right] 
     = \!\sum\limits_{n = 1}^{{N_{\rm{B}}}} {\mathbb E}\!\!\left[\!\left| \!{{ \left( \!{ {h}_{{\rm{A\!C\!A}},k,n}^{\rm{c}}}\!\right)\!}^H}\!{{x}_{l,n}}\!\right|^2\right]   \label{InterferTermReduceFirst}\\
    &\;\;\;\;\;\;\;\;\;\;\; \le  {{{\mathscr{L}}^{\rm c}_{{\rm {cas}},k}} {N\!_{\rm D}}{\overline \mu}} {\mathbb E}\!\left[\!\left\|{\boldsymbol{x}_l}\right\|^2\!\right] \!\!= \! P_0 {{{\mathscr{L}}^{\rm c}_{{\rm {cas}},k}} {N\!_{\rm D}}{\overline \mu}} ,
    \label{InterferTermReduceFirst1}
\end{alignat}
where ${ {h}_{{\rm{A\!C\!A}},k,n}^{\rm{c}}}$ and ${{x}_{l,n}}$ stand for the $n$-th element of ${\boldsymbol{h}_{{\rm{A\!C\!A}},k}^{\rm{c}}}$ and ${\boldsymbol{x}_l}$, respectively.

Substituting~\eqref{InterferTermReduceFirst1} into~\eqref{xxxxx}, \eqref{SINRDRISMath} is derived.
\end{IEEEproof}

\subsection{Impact of DRIS on Sensing}~\label{DRISJammSnesing}
Having quantified the impact of the DRIS on communication performance, we now turn to its impact on the sensing performance of the bistatic ISAC system. Specifically, the $l$-th observation vector ${{ {\boldsymbol{y}}_{\rm s}}}$ in~\eqref{SensingChannel} can be rewritten as
\begin{equation}
    {{ {\boldsymbol{y}}_l^{\rm s}}}  = \chi\!\left({\boldsymbol{h}_{{\rm d},2}^{\rm s}}\right)^{\!H} {\boldsymbol{h}_{{\rm d},1}^{\rm s}}{{\boldsymbol{x}_l}}  
    \!+\!\chi\!\left({\boldsymbol{h}_{{\rm d},2}^{\rm s}}\right)^{\!H}\! {\boldsymbol{h}_{\rm D}^{\rm s}}\!(t_{l}){{\boldsymbol{x}_l}}
     \!+   \!{{\boldsymbol{n}^{\rm s}_l}}. \label{EchoSig1}
   % & = {{\widetilde{\boldsymbol{u}}_l^{\rm s}}} + {{\widetilde{\boldsymbol{n}}_l^{\rm s}}},\label{EchoSig1}
\end{equation}

The term ${{\widetilde{\boldsymbol{n}}_l^{\rm s}}} \!= \!\chi\!\left({\boldsymbol{h}_{{\rm d},2}^{\rm s}}\right)^{\!H}\! {\boldsymbol{h}_{\rm D}^{\rm s}}\!(t_{l}){{\boldsymbol{x}_l}}
     +  {{\boldsymbol{n}^{\rm s}_l}} $ 
     consists of the DRIS-induced ACA interference and AWGN.
Furthermore, we define the overall sensing vector as 
${{\overline{\boldsymbol{y}}_{\rm s}}} = {{\overline{\boldsymbol{u}}_{\rm s}}} + {{\overline{\boldsymbol{n}}_{\rm s}}}$,
where ${{\overline{\boldsymbol{y}}_{\rm s}}} ={\rm{vec}}\!\left( \left[{{\widetilde{\boldsymbol{y}}_1^{\rm s}}},\ldots,
{{\widetilde{\boldsymbol{y}}_L^{\rm s}}}\right] \right)$,
${{\overline{\boldsymbol{u}}_{\rm s}}} =
{\rm{vec}}\!\left(\left[{{\widetilde{\boldsymbol{u}}_1^{\rm s}}},\ldots,{{\widetilde{\boldsymbol{u}}_L^{\rm s}}}\right]\right)$,
and ${{\overline{\boldsymbol{n}}_{\rm s}}} = {\rm{vec}}\!\left(\left[{{\widetilde{\boldsymbol{n}}_1^{\rm s}}},\ldots,{{\widetilde{\boldsymbol{n}}_L^{\rm s}}}\right]\right)$.
% In this section, we aim to quantify the impact of DRIS on sensing performance.
% In ISAC systems, the trade-off factor $\kappa$ in the Pareto optimization problem (P2) is typically 
% set to a small value~\cite{LiuJSAC,FLiuISAC,FLiuISAC1}.
% Moreover, the impact of the ISAC waveform resulting from the simultaneously S\&C functions has been well studied in
% previous work~\cite{MUIImpact}.
% Therefore, in the following analysis, we impose the constraint on ${\widetilde {\bf{C}}_{\bf X}}$
% by ${\widetilde {\bf{C}}_{\bf X}}  =  \frac{P_0}{N_{\rm B}}{\bf I}_{N\!_{\rm B}}$ 
% to clearly highlight the impact of the DRIS on sensing performance.
% Note that the sensing performance under the strict sensing constraint serves as its lower bound.
We assume that the parameters to be estimated are limited to the AoD and AoA, 
i.e., ${\boldsymbol{\theta }} =  [\theta_1,\theta_2 ]^T$.
Before characterizing the impact of the DRIS on the sensing performance, 
we first present the statistical properties of the  cascaded DRIS-based sensing path ${\boldsymbol{h}_{\rm D}^{\rm s}}(t)$ 
in Proposition~\ref{Proposition2}.
\begin{proposition}
    \label{Proposition2}
    The i.i.d. elements of ${\boldsymbol{h}_{\rm D}^{\rm s}}\!(t)$ converge in distribution to $\mathcal{CN}\!\left( {0,  {{{\mathscr{L}}^{\rm s}_{{\rm {cas}}}} {N\!_{\rm D}} {\overline \nu} } } \right)$ as $N_{\rm D} \to \infty$, i.e.,
    \begin{equation}
        { { {h}_{{\rm D},n}^{\rm s}}\!(t) } \mathop  \to \limits^{\rm{d}} \mathcal{CN}\!\!\left( {0,  {{{\mathscr{L}}^{\rm s}_{{\rm {cas}}}} {N\!_{\rm D}}{\overline \nu} } } \right), 
        n = 1,\ldots, N_{\rm B},
        \label{hDsSta}
    \end{equation}
where  
$\overline \nu = {\sum\nolimits_{i = 1}^{{2^b}} {{p_i}{\mu^2_i}} }$.
\end{proposition}

\begin{IEEEproof}
    See Appendix~\ref{AppendixB}.
\end{IEEEproof}

To quantify the impact of the DRIS on the sensing performance,
we derive the Fisher information matrix (FIM)~\cite{CRLBbook} in the presence of a DRIS 
as stated in Theorem~\ref{Theorem2}.

\begin{theorem}
    \label{Theorem2}
   For the bistatic ISAC system in the presence of a DRIS,  
    the FIM corresponding to  ${\boldsymbol{\theta }} =  [\theta_1,\theta_2 ]^T$ 
    obtained by 
    using the given ISAC waveform ${\bf X}$ is derived as
    \begin{equation}
        \left[{\bf F}\right]_{i,j} =\!  \sum\limits_{l = 1}^L{\!2 \Re\!\! \left\{ \!\frac{{\partial {{{ \left( {\widetilde{\boldsymbol{u}}_l^{\rm s}} \right)\!}^H}}}}{{\partial {\theta _i}}}
        \!{\widetilde {\bf {R}}}^{\!-\!1}_{l} \!\frac{{\partial {{\widetilde{\boldsymbol{u}}_l^{\rm s}}}}}{{\partial {\theta _j}}}\!\right\}}
        \!+\! {\sum\limits_{l = 1}^L \!{\rm{tr}}\!\!\left\{\!{\widetilde {\bf {R}}}^{\!-\!1}_{l}\frac{{\partial { {\widetilde {\bf {R}}}_{l}}}}{{\partial {\theta _i}}}
        {\widetilde {\bf {R}}}^{\!-\!1}_{l} \frac{{\partial {\widetilde {\bf {R}}}_{l}}}{{\partial {\theta _j}}}\!\right\}},
         \label{FIM}
    \end{equation}
where 
\begin{equation}
    {\widetilde {\bf {R}}}_l = \!\chi^2 \!{{\mathscr{L}}^{\rm s}_{\!{\rm {d}},2}}{{\mathscr{L}}^{\rm s}_{\!{\rm {cas}}}} {N\!_{\rm D}}{\overline \nu} 
    {\left\|{\boldsymbol{x}}_{l}\right\|^2}\!\!
    \left(  { \boldsymbol{\alpha}}^H_{{\!N\!_{\rm S}}} \!\!\left({\theta_2}\right) \!{ \boldsymbol{\alpha}}_{{\!N\!_{\rm S}}} \!\! \left({\theta_2}\right) \! \right)
       \!+\! {\delta^2_{\rm s}}{\bf I}\!_{N_{\rm S}}.
   \label{CovRDR1}
\end{equation}
\end{theorem}

\begin{IEEEproof}
According to  Proposition~\ref{Proposition2}, 
the covariance matrix of ${{\overline{\boldsymbol{n}}_{\rm s}}}$
can be calculated from
\begin{alignat}{1}%{alignat}{1}
    {\overline {\bf {R}}} _{\! \rm D} 
    &=  {\mathbb{E} }\! \left[ \left({{\overline{\boldsymbol{n}}_{\rm s}}} -
    {\mathbb{E}}\!\left[{{\overline{\boldsymbol{n}}_{\rm s}}}\right] \right)\!{\left({{\overline{\boldsymbol{n}}_{\rm s}}} -
    {\mathbb{E}}\!\left[{{\overline{\boldsymbol{n}}_{\rm s}}}\right] \right) }^{H}  \right] \label{CovRD1R1}\\
    & =  {\rm{blkdiag}}\!\left(\!{\widetilde {\bf {R}}}_1 ,\ldots,
    {\widetilde {\bf {R}}}_L \!\right),
    \label{CovRD1R12}
\end{alignat}
where
\begin{equation} 
    {\widetilde {\bf {R}}}_l  =
     {\mathbb{E} }\! \left[ \left({{\widetilde{\boldsymbol{n}}_l^{\rm s}}} -
    {\mathbb{E}}\!\left[{{\widetilde{\boldsymbol{n}}_l^{\rm s}}}\right] \right)\!
    {\left({{\widetilde{\boldsymbol{n}}_l^{\rm s}}} -
    {\mathbb{E}}\!\left[{{\widetilde{\boldsymbol{n}}_l^{\rm s}}}\right] \right) }^{H}  \right].
    \label{CovRD1R12V1} 
\end{equation} 
% In~\eqref{CovRD1R12V1}, we exploit a small trick to simplify the expression, i.e., ${\rm{vec}}\!\left(\!{\bf A}{\bf B}{\bf C}\right) = \left({\bf C}^T\!\otimes\!{\bf A}\right) 
% \!{\rm{vec}}\!\left({\bf B} \right)$. 
Consequently, 
we have
\begin{equation} 
    {\widetilde {\bf {R}}}_l =
     {\mathbb{E} }\!\! \left[ {{\widetilde{\boldsymbol{n}}_l^{\rm s}}} 
    {\left({{\widetilde{\boldsymbol{n}}_l^{\rm s}}}  \right) }^{\!H}  \!\right] \!\!+\!   {\mathbb{E} }\!\! \left[\! \chi^2\!\left({\boldsymbol{h}_{{\rm d},2}^{\rm s}}\right)^{\!H}\! {\boldsymbol{h}_{\rm D}^{\rm s}}\!(t_{l}){{\boldsymbol{x}_l}}
      {{\boldsymbol{x}^H_l}}\! \left({\boldsymbol{h}_{\rm D}^{\rm s}}\!(t_{l})\!\right)^{\!H} 
     \! {\boldsymbol{h}_{{\rm d},2}^{\rm s}} \!\right].
    \label{CovRD1R12V13} 
\end{equation} 
Consequently, \eqref{CovRDR1} is then obtained.

Moreover, based on Proposition~\ref{Proposition2},
${{\overline{\boldsymbol{y}}_{\rm s}}}$ follows a complex Gaussian distribution,
i.e., 
\begin{alignat}{1}%{alignat}{1}
    \nonumber
    {{\overline {\boldsymbol{y}}_{\rm s}}} & \sim \mathcal{CN}\!\left({{\overline {\boldsymbol{u}}_{\rm s}}}, {\overline {\bf {R}}} _{\! \rm D}  \right) \\
    & = \frac{1}{{{\pi ^{ {N_{\rm{S}}}L}}\det  \!\left( {\overline {\bf {R}}} _{\! \rm D}  \right)}}\exp \!\left(\! -
    \left({{\overline{\boldsymbol{y}}_{\rm s}}} \!-\! {{\overline{\boldsymbol{u}}_{\rm s}}}\right)^H{{\overline {\bf {R}}}^{-\!1}_{\! \rm D} }
    \!\left({{\overline{\boldsymbol{y}}_{\rm s}}} \!-\! {{\overline{\boldsymbol{u}}_{\rm s}}}\right)\! \right),
    \label{CNOb1}
\end{alignat}
where $\det ( {\overline{\bf {R}}} _{\! \rm D})$ is the determinant of ${\overline{\bf {R}}} _{\! \rm D}$.
Consequently, the log-likelihood function of ${{\overline {\boldsymbol{y}}_{\rm s}}}$ 
is expressed as
\begin{alignat}{1}
    & \mathcal{L}\!\left({\boldsymbol{\theta }}\right)   = 
     -\ln\det({\overline{\bf {R}}} _{\! \rm D}) -
     \left({{\overline{\boldsymbol{y}}_{\rm s}}} - {{\overline{\boldsymbol{u}}_{\rm s}}}\right)^H
     {\overline{\bf {R}}}^{ -\!1}_{\! \rm D}
    \left({{\overline{\boldsymbol{y}}_{\rm s}}} - {{\overline{\boldsymbol{u}}_{\rm s}}}\right)
    \label{Loglikelihoody}\\
    & \;\; =- \! {\sum\limits_{l = 1}^L {\ln \det\! \left( {{{\widetilde {\bf{R}}}_l}} \right)}}\! - \!\!
    {\sum\limits_{l = 1}^L \!\left({{\widetilde{\boldsymbol{y}}_l^{\rm s}}} \!-\! {{\widetilde{\boldsymbol{u}}_l^{\rm s}}}\right)^H{{{{\widetilde {\bf{R}}}_l}}}^{\!-\!1} 
    \!\left({{\widetilde{\boldsymbol{y}}_l^{\rm s}}} \!-\! {{\widetilde{\boldsymbol{u}}_l^{\rm s}}}\right)}, \label{LoglikelihoodyA1}
\end{alignat}
where constant terms are ignored for brevity.

Given the log-likelihood function $\mathcal{L}\!\left({\boldsymbol{\theta }}\right)$,
the FIM can be computed by 
\begin{equation}
    \left[{\bf F}\right]_{i,j} = {\mathbb{E} }\!\left[\!\frac{{\partial \mathcal{L}\!\left({\boldsymbol{\theta }}\right)}}{{\partial {\theta _i}}}
    \frac{{\partial \mathcal{L}\!\left({\boldsymbol{\theta }}\right)}}{{\partial {\theta _j}}}\!\right],
    \label{FIM1stTerm} 
\end{equation}
where $i,j =1,2$.

Substituting~\eqref{LoglikelihoodyA1} into~\eqref{FIM1stTerm}, the FIM in~\eqref{FIM}
is then obtained.
Based on the definitions of ${\boldsymbol{h}_{{\rm d},1}^{\rm s}}$ and ${\boldsymbol{h}_{{\rm d},2}^{\rm s}}$
in~\eqref{hds} and~\eqref{hd2s}, 
the partial derivative of ${ {{{\widetilde{\boldsymbol{u}}_l^{\rm s}}}}}$ 
with respect to ${\theta _1}$ in~\eqref{FIM} can be further expressed as 
\begin{equation} %{alignat}{1}
    \frac{{\partial {{{\widetilde{\boldsymbol{u}}_l^{\rm s}}}}}}{{\partial {\theta _1}}} 
    =\chi \sqrt{{{\mathscr{L}}^{\rm s}_{{\rm {d}},1}}{{\mathscr{L}}^{\rm s}_{{\rm {d}},2}}} {  { \boldsymbol{\alpha}}^H_{{\!N_{\!\rm S}}} \!\!\left({\theta_2}\!\right) }
    \frac{ {\partial { { \boldsymbol{\alpha}}_{{\!N_{\!\rm B}}} \!\!\left({\theta_1}\!\right)}}}{{\partial {\theta _1}}} {\boldsymbol{x}}_{l},
    \label{Partialu}
\end{equation}
where
\begin{equation} %{alignat}{1}
    \frac{ {\partial { { \boldsymbol{\alpha}}_{{\!N_{\!\rm B}}} \!\!\left({\theta_1}\!\right)}}}{{\partial {\theta _1}}}
    = j2\pi \Delta \cos{\theta_1}  { { \boldsymbol{\alpha}}_{{\!N_{\!\rm B}}} \!\!\left({\theta_1}\!\right)}{\bf \Lambda }_{N_{\rm B}},
    \label{PartialULA1}
\end{equation}
and ${\bf \Lambda }_{N_{\rm B}}\! = {\rm{diag}}\!\left(0,1,\ldots,(N_{\rm B}-1) \right)$.
Furthermore, the partial derivative of ${ { {{{\widetilde{\boldsymbol{u}}_l^{\rm s}}}}}}$ 
with respect to ${\theta _2}$ in~\eqref{FIM}
is calculated as
\begin{equation} %{alignat}{1}
    \frac{{\partial { {{{\widetilde{\boldsymbol{u}}_l^{\rm s}}}}}}}{{\partial {\theta _2}}}  
    =\chi \sqrt{{{\mathscr{L}}^{\rm s}_{{\rm {d}},1}}{{\mathscr{L}}^{\rm s}_{{\rm {d}},2}}}
     \frac{ {\partial { { \boldsymbol{\alpha}}^H_{{\!N_{\!\rm S}}} \!\!\left({\theta_2}\!\right)}}}{{\partial {\theta _2}}}
     \!{ \boldsymbol{\alpha}}_{{\!N_{\!\rm B}}} \!\!\left({\theta_1}\!\right) {\boldsymbol{x}}_{l},
    \label{Partialu2}
\end{equation}
where
\begin{equation} %{alignat}{1}
    \frac{ {\partial { { \boldsymbol{\alpha}}^H_{{\!N_{\!\rm S}}} \!\!\left({\theta_2}\!\right)}}}{{\partial {\theta _2}}}
    = -j2\pi \Delta \cos{\theta_2}{\bf \Lambda }_{N_{\rm S}} { { \boldsymbol{\alpha}}^H_{{\!N_{\!\rm S}}} \!\!\left({\theta_2}\!\right)}.
    \label{PartialULA2}
\end{equation}

According to the Sherman--Morrison formula, the inverse matrix of ${\widetilde {\bf {R}}}_{l}$ in~\eqref{FIM}
can be calculated as
\begin{equation} %{alignat}{1}
    { {\widetilde {\bf {R}}}^{-1}_{l}}=\frac{1}{{\delta^2_{\rm s}}} {\bf I}_{\!{{N}_{\rm S}}} 
    \!- 
    \frac{\chi^2 \!{{\mathscr{L}}^{\rm s}_{\!{\rm {d}},2}}{{\mathscr{L}}^{\rm s}_{\!{\rm {cas}}}} {N\!_{\rm D}}{\overline \nu} {\left\|{\boldsymbol{x}}_{l}\right\|^2}
    { \boldsymbol{\alpha}}^H_{{\!N_{\rm S}}} \!\!\left({\theta_2}\right) { \boldsymbol{\alpha}}_{{\!N_{\rm S}}} \!\!\left({\theta_2}\right)}
    {{\delta^4_{\rm s}} + \chi^2 \!{{\mathscr{L}}^{\rm s}_{\!{\rm {d}},2}}{{\mathscr{L}}^{\rm s}_{\!{\rm {cas}}}} {N\!_{\rm D}}{\overline \nu}  {\left\|{\boldsymbol{x}}_{l}\right\|^2}{{{N}_{\rm S}}}{\delta^2_{\rm s}}}.
    \label{Rlinv}
\end{equation}
 
Finally, the partial derivative of ${\widetilde {\bf {R}}}_{l}$ with respect to ${\theta _2}$ 
in~\eqref{FIM} can be calculated from~\eqref{CovRDR1} as
\begin{equation} %{alignat}{1}
     \frac{ {\partial { {\widetilde {\bf {R}}} _{l}}}}{{\partial {\theta _2}}} =   
     \chi^2 \!{{\mathscr{L}}^{\rm s}_{\!{\rm {d}},2}}{{\mathscr{L}}^{\rm s}_{\!{\rm {cas}}}} {N\!_{\rm D}}{\overline \nu} 
    {\left\|{\boldsymbol{x}}_{l}\right\|^2}\!
    \left( \!\frac{ {\partial {{ \boldsymbol{\alpha}}^H_{{\!N_{\rm S}}} \!\!\left({\theta_2}\right) { \boldsymbol{\alpha}}_{{\!N_{\rm S}}} \!\!\left({\theta_2}\right)}}}{{\partial {\theta _2}}}
    \right).
    \label{PartialRD2}
\end{equation}
In~\eqref{PartialRD2}, the partial derivative  can be further expressed as
\begin{equation} %{alignat}{1}
    \frac{ {\partial {{ \boldsymbol{\alpha}}^H_{{\!N_{\rm S}}} \!\!\left({\theta_2}\right) { \boldsymbol{\alpha}}_{{\!N_{\rm S}}} \!\!\left({\theta_2}\right)}}}{{\partial {\theta _2}}}
    = j2\pi \Delta \cos{\theta_2}{\bf D }_{N_{\rm S}} \!\odot\! {{ \boldsymbol{\alpha}}^H_{{\!N_{\rm S}}} \!\!\left({\theta_2}\right) { \boldsymbol{\alpha}}_{{\!N_{\rm S}}} \!\!\left({\theta_2}\right)},
    \label{PartialULAULA2}
\end{equation}
where 
$\odot$ represents the Hadamard product, 
and ${\bf D }_{N_{\rm S}} \in {\mathbb{R}}^{N_{\rm S} \times N_{\rm S}}$ is the matrix whose element in the $m$-th row and
the $n$-th column is $(n-m)$.
It is worth noting that the covariance matrix ${\widetilde {\bf {R}}}_{l}, \forall l$
is independent of the parameter $\theta_1$, and thus its partial derivative
with respect to $\theta_1$  is equal to zero.
\end{IEEEproof}

Unlike existing DRIS works in wireless communications, 
the DRIS exerts a distinctive impact on the sensing performance of bistatic ISAC systems.
 According to Theorem~\ref{Theorem2},  
the DRIS 
degrades the estimation accuracy of the AoD while, interestingly, improving that of the AoA.
In this work, we further use the CRLB to strictly quantify the impact of the DRIS on 
the sensing performance in the bistatic ISAC system.
Based on Theorem~\ref{Theorem2}, the CRLB of the AoD $\theta_1$ and AoA $\theta_2$ estimates
can be obtained from the inverse of $\bf F$~\cite{CRLBbook}, i.e., 
\begin{alignat}{1}
     { {C\!R\!L\!B}}\!\left(\theta_1\right) &= \left[{\bf F}^{-1}\right]_{1,1} 
     = \frac{\left[{\bf F} \right]_{2,2}}{\det\!\left({\bf F}\right)}, \label{CRLB1}\\
     { {C\!R\!L\!B}}\!\left(\theta_2\right) &= \left[{\bf F}^{-1}\right]_{2,2} = 
     \frac{\left[{\bf F} \right]_{1,1}}{\det\!\left({\bf F}\right)}, \label{CRLB2}
\end{alignat}
where ${\det\!\left({\bf F}\right)}$ represents the determinant of $\bf F$.

Note that, in the absence of DRIS, the FIM in Theorem~\ref{Theorem2} reduces to
\begin{equation}
    \left[{\bf F}\right]_{i,j} =  \frac{1}{{\delta^2_{\rm s}} }\sum\limits_{l = 1}^L{\!2 
    \Re\!\! \left\{ \!\frac{{\partial {{{ \left( {\widetilde{\boldsymbol{u}}_l^{\rm s}} \right)\!}^H}}}}{{\partial {\theta _i}}}
          \!\frac{{\partial {{\widetilde{\boldsymbol{u}}_l^{\rm s}}}}}{{\partial {\theta _j}}}\!\right\}}.
     \label{FIMWO}
\end{equation}

% In the traditional target detection, 
% the two-dimensional (2D)-MUSIC algorithm~\cite{BistaticRadar,DectP} 
% is widely employed to estimate $\theta_1$ and $\theta_2$.  
% However,  
% In the classical two-dimensional (2D)-MUSIC algorithm~\cite{BistaticRadar,DectP}, the estimation accuracy 
% is fundamentally limited by the resolution of the two-dimensional grid search.
% However, using a finer grid resolution leads to a significantly higher computational complexity.
% To address this issue, 
In practice, maximum likelihood estimation (MLE) is often used to accurately estimate the AoD and AoA~\cite{CRLBbook}.
Under standard regularity conditions, 
the MLE is asymptotically efficient and attains
the CRLBs as the number of snapshots or the signal-to-noise ratio (SNR) increases.
% Mathematically, the MLEs of the AoD and the AoA can be described as follows.
Based on Proposition~\ref{Proposition2},  ${{\overline{\boldsymbol{y}}_{\rm s}}}$ 
follows a complex Gaussian distribution, as given in~\eqref{CNOb1}.
Therefore, from the log-likelihood function ${\mathcal{L} \!\left({\boldsymbol{\theta }}\right) }$ in~\eqref{Loglikelihoody}, 
we then derive its gradient $\nabla {\mathcal{L} \!\left({\boldsymbol{\theta }}\right) } = 
\left[\frac{{\partial {\mathcal{L} \!\left({\boldsymbol{\theta }}\right) }}}{{\partial {\theta _1}}},
\frac{{\partial {\mathcal{L} \!\left({\boldsymbol{\theta }}\right) }}}{{\partial {\theta _2}}}\right]^T$ 
by differentiating ${\mathcal{L} \!\left({\boldsymbol{\theta }}\right) }$.
More specifically, the partial derivative of ${\mathcal{L} \!\left({\boldsymbol{\theta }}\right) }$ 
with respect to $\theta_1$ is given by
\begin{equation} %{alignat}{1}
    \frac{{\partial {\mathcal{L} \!\left({\boldsymbol{\theta }}\right) }}}{{\partial {\theta _1}}} = 
    \sum\limits_{l = 1}^L{2\Re\!\left\{ \frac{{\partial {{{ \left( {\widetilde{\boldsymbol{u}}_l^{\rm s}} \right)\!}^H}}}}{{\partial {\theta _1}}} 
    {\widetilde {\bf {R}}}^{\!-\!1}_{l}
    \left({{\widetilde{\boldsymbol{y}}_l^{\rm s}}} - {{\widetilde{\boldsymbol{u}}_l^{\rm s}}}\right)\right\}}.
    \label{PDL1}
\end{equation}

Based on~\eqref{CovRDR1}, ${\widetilde {\bf {R}}}_{l}$
is a function of $\theta_2$. Therefore, the partial derivative of ${\mathcal{L} \!\left({\boldsymbol{\theta }}\right) }$ 
with respect to $\theta_2$ is computed as
\begin{alignat}{1}
    \nonumber
    &\frac{{\partial {\mathcal{L} \!\left({\boldsymbol{\theta }}\right) }}}{{\partial {\theta _2}}} = 
    \sum\limits_{l = 1}^L{2\Re\!\left\{ \frac{{\partial {{{ \left( {\widetilde{\boldsymbol{u}}_l^{\rm s}} \right)\!}^H}}}}{{\partial {\theta _2}}} 
    {\widetilde {\bf {R}}}^{\!-\!1}_{l}
    \left({{\widetilde{\boldsymbol{y}}_l^{\rm s}}} - {{\widetilde{\boldsymbol{u}}_l^{\rm s}}}\right)\right\}}  + \\
    & \sum\limits_{l = 1}^L{\left({{\widetilde{\boldsymbol{y}}_l^{\rm s}}} - {{\widetilde{\boldsymbol{u}}_l^{\rm s}}}\right)^H 
   {\widetilde {\bf {R}}}^{\!-\!1}_{l} \frac{{\partial { {\widetilde {\bf {R}}} _{l}}}}{{\partial {\theta _2}}} 
   {\widetilde {\bf {R}}}^{\!-\!1}_{l} \left({{\widetilde{\boldsymbol{y}}_l^{\rm s}}} -  {{\widetilde{\boldsymbol{u}}_l^{\rm s}}}\right)} \!- \!\!\sum\limits_{l = 1}^L{{\rm{tr}}\!\!\left\{ {\widetilde {\bf {R}}}^{\!-\!1}_{l} \frac{{\partial { {\widetilde {\bf {R}}} _{l}}}}{{\partial {\theta _2}}}  \right\}},
     \label{PDL2}  
\end{alignat}
where the following matrix calculus identities are utilized:
\begin{equation} %{alignat}{1}
    \frac{{\partial { \ln\det({\widetilde {\bf {R}}} _{l}) }}}{{\partial {\theta _2}}} = 
    {\rm{tr}}\!\left\{ {\widetilde {\bf {R}}}^{\!-\!1}_{l}  
    \frac{{\partial { {\widetilde {\bf {R}}} _{l} }}}{{\partial {\theta _2}}} \right\},
    \label{PDL2A1}
\end{equation}
and
\begin{equation} %{alignat}{1}
    \frac{{\partial {{\widetilde {\bf {R}}}^{\!-\!1}_{l} }}}{{\partial {\theta _2}}} = 
    -{{\widetilde {\bf {R}}}^{\!-\!1}_{l} } \frac{{\partial {{\widetilde {\bf {R}}} _{l} }}}{{\partial {\theta _2}}}
    {{\widetilde {\bf {R}}}^{\!-\!1}_{l} }.
    \label{PDL2A2}
\end{equation}

The optimal AoA and AoD estimates are obtained by maximizing the log-likelihood function ${\mathcal{L} \!\left({\boldsymbol{\theta }}\right) }$.
Mathematically,
they are iteratively updated using the gradient ascent method 
\begin{equation} %{alignat}{1}
     {{\boldsymbol{\theta }}}^{(i+1)} = {{\boldsymbol{\theta }}}^{(i)} + 
     \zeta {\left.\nabla {\mathcal{L} \!\left({\boldsymbol{\theta }}\right) } \right|_{{{\boldsymbol{\theta }}}^{(i)}}},
    \label{MLEGAM}
\end{equation}
until
\begin{equation} %{alignat}{1}
     \left\|{{\boldsymbol{\theta }}}^{(i+1)} - {{\boldsymbol{\theta }}}^{(i)}\right\|^2 \le \sigma ,
    \label{MLEGAMThe}
\end{equation}
where $i \ge 1$ represents the $i$-th iteration,
$\zeta>0$ denotes the learning rate that controls the update step size, 
and $\sigma$ is a predefined convergence threshold.
% It is worth noting that the initial value in MLE can be obtained using methods such as the 2D-MUSIC algorithm in~\eqref{AOAAODMUSIC}, 
% which accelerates the convergence process and increases the likelihood of finding the correct solution.

\section{Simulation Results and Discussion}\label{ResDis}
In this section, we provide numerical results to evaluate the impact of 
the DRIS on bistatic ISAC systems and validate the accuracy of the derived theoretical analysis
in Section~\ref{ISACDesign}. 
We consider a bistatic ISAC system where the ISAC BS 
is equipped with an 8-element ULA located at (0m, 0m, 3m)
to communicate with 4 single-antenna users 
that are randomly distributed in the circular region $S$ centered at (0m, 180m, 0m) with a radius of 20m. 
Meanwhile, the DRIS with 4096 ($N_{{\rm D},h} = 64,N_{{\rm D},v} = 64 $) passive reflective elements is deployed at (-1m, 0m, 2.5m).
We assume that the DRIS  has a constant amplitude and one-bit quantized phase shifts  taken from $\Psi = \{\frac{2\pi}{5},\frac{7\pi}{5}\}$,
while the two phase shifts are chosen with equal probability. 
Consequently, $\overline \mu$ in Theorem~\ref{Theorem1} and $\overline \nu $ 
in Theorem~\ref{Theorem2} are calculated as 1.
In addition, the length of the data frame is $L = 80$,
and the trade-off factor in~$\left( {{\rm{P}}2} \right)$ is $\kappa = 0.2$~\cite{FLiuISAC}.
Furthermore, the sensing receiver in the bistatic ISAC system, located at (0m, 60m, 0m), 
is also equipped with an 8-element ULA 
to detect the sensing target, whose reflection cross-section $\chi$ is 0.9.
Specifically, the target is positioned 20m from the origin, 
with its bearing drawn uniformly between $\frac{\pi}{6}$ and $\frac{\pi}{3}$
to either side of the $x$-axis.

\begin{table}
    \footnotesize
    \centering
    \caption{Wireless Channel Simulation Parameters}
    \label{tab1}
    \begin{threeparttable}
    \begin{tabular}{ c|c }
    \hline
    Parameter        &Value\\
    \hline
    Large-scale LoS fading       & $ 35.6 + 22{\log _{10}}({d}) $ (dB) \\
    \hline
    Large-scale NLoS fading      &$32.6+36.7{\log _{10}}({d})$ \\
    \hline
    \end{tabular}
    \end{threeparttable}
\end{table} 

Based on the settings above, the wireless channel $\bf G$ is constructed using a
near field channel model, as expressed in~\eqref{Ricianchan} and~\eqref{GLOS}. 
Furthermore, the communication channels ${\bf H}_{\rm I}^{\rm c}$ and ${\bf H}_{\rm d}^{\rm c}$ 
are both based on far field channel models in~\eqref{HIkeq} and~\eqref{Hdkeq}.
In addition, the sensing channels ${\boldsymbol{h}}^{\rm s}_{{\rm{d}},1}$, ${\boldsymbol{h}}^{\rm s}_{{\rm{d}},2}$, 
and ${\boldsymbol{h}}^{\rm s}_{{\rm{I}}}$ are built based on~\eqref{hds}, \eqref{hd2s}, and \eqref{hIs}.
The large-scale  LoS  and NLoS channel fading coefficients  are defined in Table~\ref{tab1} based on 3GPP propagation models~\cite{3GPP}, 
and the variance of the communication noise is $\sigma^2_{\rm c}\!=\!-170\!+\!10\log_{10}\left(BW\right)$ dBm with a transmission bandwidth of 180 KHz.
For simplicity, we assume that the variance of sensing noise $\sigma^2_{\rm s}$ is equal to $\sigma^2_{\rm c}$.

\begin{figure}[!t]
    \centering
    \includegraphics[scale=0.63]{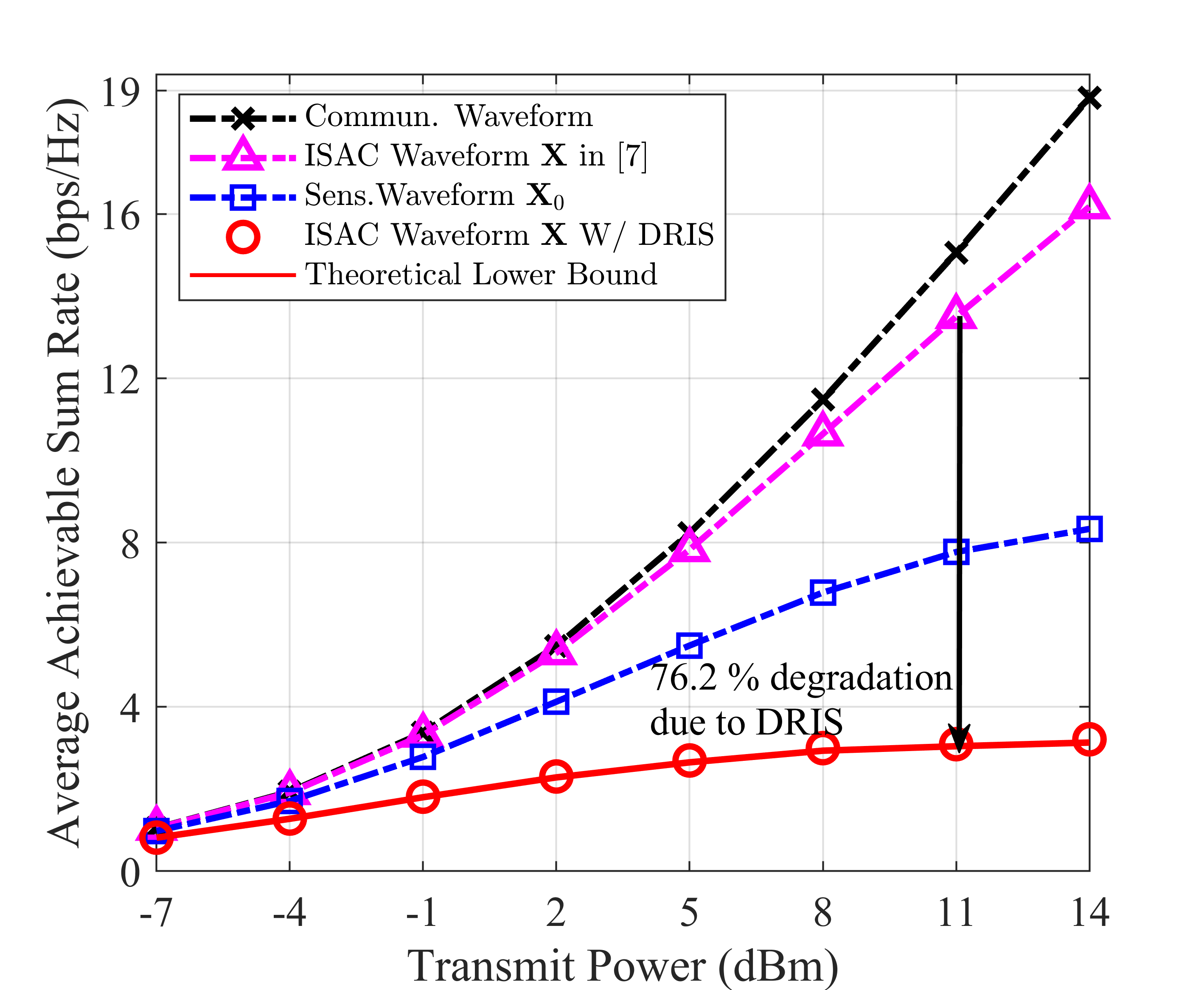}
    \caption{Average achievable sum rate vs. total transmit power for different benchmarks.}
    \label{ResfigCPpower}
\end{figure}
Fig.~\ref{ResfigCPpower} illustrates 
the achievable sum rates versus the transmit power at the ISAC BS evaluated under the following approaches: 
1) the best sum rates without MU interference and ACA interference (Commun. Waveform);
2) the sum rates achieved by a traditional ISAC waveform in~\cite{FLiuISAC}, i.e., $\bf{X}$
obtained from $\left( {{\rm{P}}2}-{{\rm{E}}1} \right)$ with $\kappa =0.2$ (ISAC Waveform $\bf X$ in~\cite{FLiuISAC});
3) the sum rates achieved by the optimal sensing waveform ${\bf X}_0$ 
with  the strict equality constraint in~\eqref{ISACWaveLim} (Sens. Waveform ${\bf X}_0$);
4) the sum rates achieved by the ISAC waveform $\bf X$ in the presence of a DRIS (ISAC 
Waveform $\bf X$ W/ DRIS); and
5) the corresponding theoretical analysis derived in Theorem~\ref{Theorem1} (Theoretical Lower Bound).

\begin{figure*}[!t]
    \centering
    \subfloat{
            \includegraphics[scale=0.62]{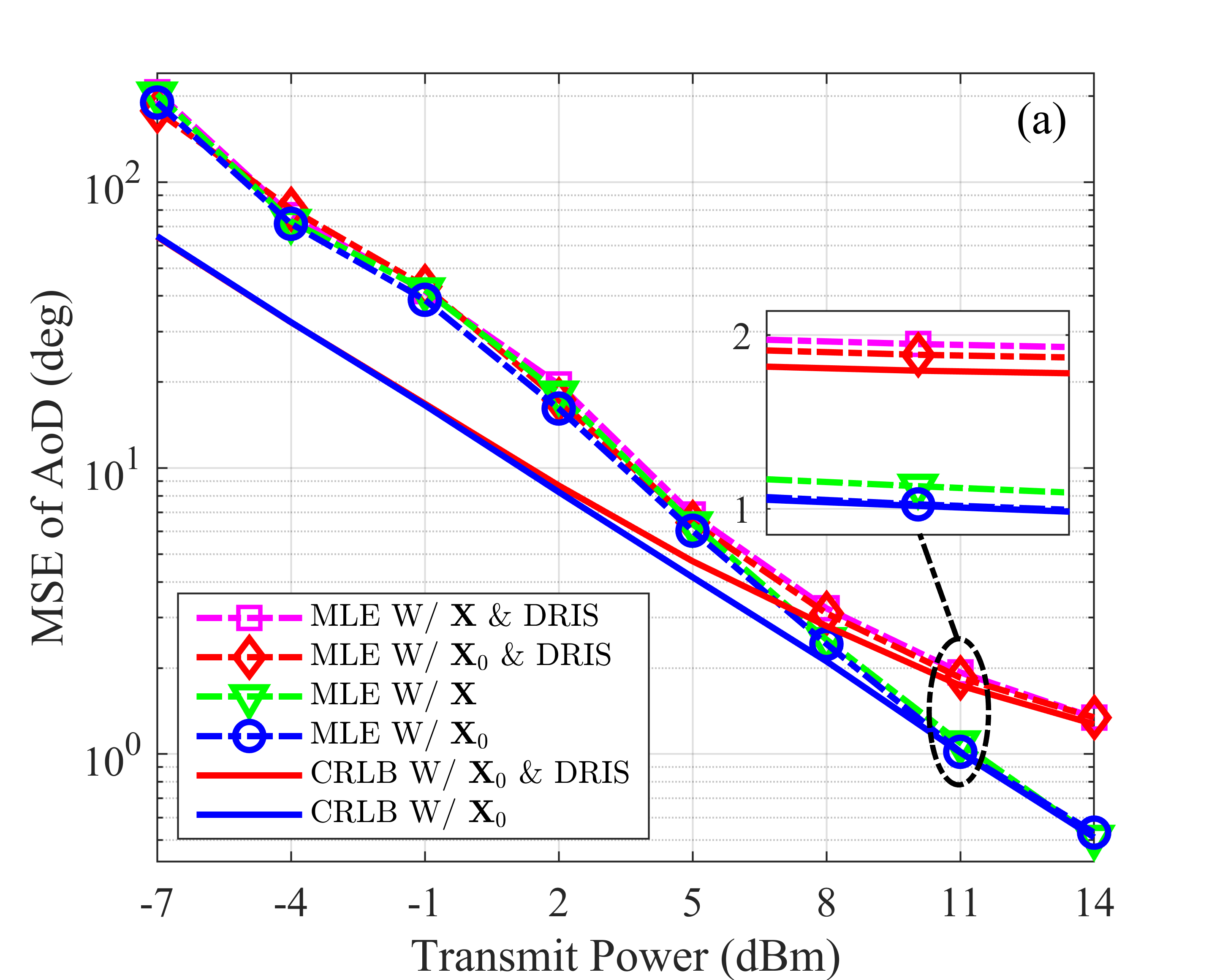}}\hspace{25pt}
    \subfloat{
            \includegraphics[scale=0.62]{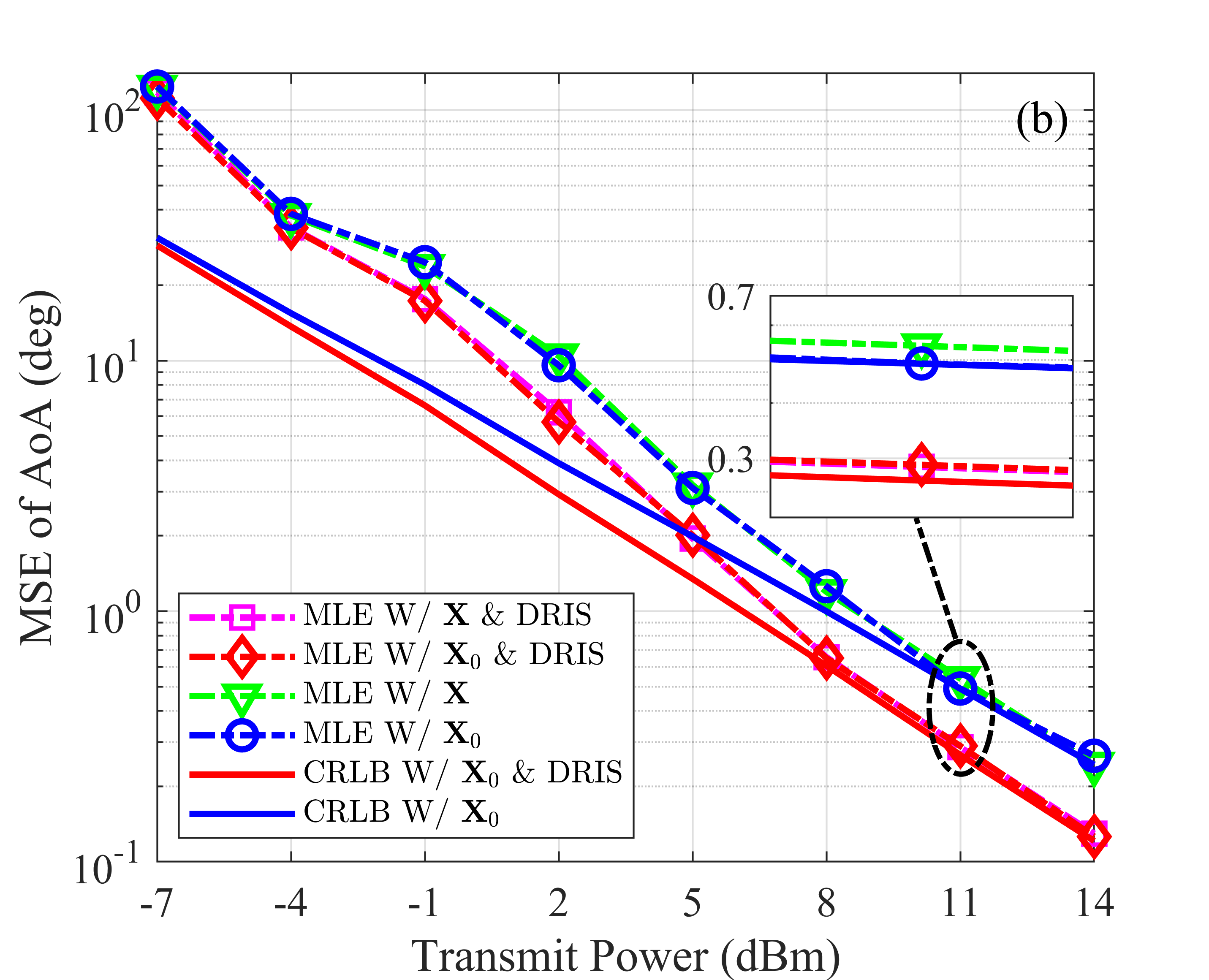}}
   \caption{Mean Square Error (MSE)  of (a)  AoD ($\theta_1$) and (b) AoA ($\theta_2$) estimates versus total transmit power
for different ISAC waveforms under MLE.}
    \label{ResFigSensingvsPowerMLE}
\end{figure*} 
It can be observed from Fig.~\ref{ResfigCPpower} that the ISAC waveform ${\bf X}$, 
obtained from $\left({\rm P2}-{\rm E1}\right)$, 
can effectively strike a trade-off between communication and sensing performance, 
as compared with the optimal sensing waveform ${\bf X}_0$.
Nevertheless, the results clearly show that 
the DRIS-induced ACA interference considerably reduces achievable sum rates.
Intuitively, increasing the total transmit power $P_0$ at the ISAC BS 
is expected to improve the achievable sum rates and mitigate the ACA interference imposed by the DRIS.
However, one can see that a higher transmit power does not effectively alleviate the adverse impact of the DRIS 
on the sum rates; 
instead, as indicated by Theorem~\ref{Theorem1}, 
it leads to an exacerbation of the DRIS-induced ACA interference.
More specifically, at 11 dBm total transmit power, 
the presence of the DRIS results in a degradation of 76.2\% 
in the achievable sum rates using the ISAC waveform $\bf X$.
Moreover, from Fig.~\ref{ResfigCPpower}, 
it can be seen that 
the simulation results are in close agreement with the theoretical derivations presented in Theorem~\ref{Theorem1}.

Fig.~\ref{ResFigSensingvsPowerMLE}
presents the sensing performance via the MLE described in Section~\ref{DRISJammSnesing}, 
together with the CRLB derived in Theorem~\ref{Theorem2}.
Specifically, we illustrate the mean square error (MSE) using the MLE method
and the CRLB 
achieved by the following waveforms:
1) the MSE achieved by the ISAC waveform $\bf X$ with DRIS (MLE W/ $\bf X$ \& DRIS);
2) the MSE achieved by the optimal sensing waveform ${\bf X}_0$ with DRIS (MLE W/ ${\bf X}_0$ \& DRIS);
3) the MSE achieved by the ISAC waveform $\bf X$ without DRIS (MLE W/ $\bf X$);
4) the MSE achieved by the optimal sensing waveform ${\bf X}_0$ without DRIS (MLE W/ ${\bf X}_0$);
5) the CRLB using ${\bf X}_0$ and in the presence of a DRIS (CRLB W/ ${\bf X}_0$ \& DRIS);
and
6) the CRLB  using ${\bf X}_0$ and without DRIS (CRLB W/ ${\bf X}_0$).
Fig.~\ref{ResFigSensingvsPowerMLE}~(a) presents the MSE and CRLB results for the AoD (i.e., \(\theta_1\)) estimation, 
while Fig.~\ref{ResFigSensingvsPowerMLE}~(b) shows the MSE and CRLB results for the AoA (i.e., \(\theta_2\)) estimation.
Moreover,
the MSE in Fig.~\ref{ResFigSensingvsPowerMLE} is defined as
\begin{equation} %{alignat}{1}
    {\rm{MSE}} =   {\frac{1}{N}\sum\limits_{n = 1}^N {{{\left( {{\theta _n} - 
    {{\overline \theta  }_n}} \right)}^2}} } ,
    \label{RMSE}
\end{equation}
where $N$ is the number of samples, 
${\theta _n}$ represents the $n$-th true value, 
and ${{\overline \theta}_n}$ denotes the $n$-th predicted value, 
respectively. 

From Fig.~\ref{ResFigSensingvsPowerMLE},
the MLE curves closely approach their CRLBs as the transmit power $P_0$ increases, 
validating the asymptotic efficiency of the MLE,
i.e., its MSE converges to the CRLB in the high SNR regime so that no other unbiased estimator can achieve a lower variance.
The performance gap between the ISAC waveform ${\bf X}$ and the optimal sensing waveform 
${\bf X}_0$ 
remains small in both cases, 
indicating that the ISAC waveform ${\bf X}$ preserves most of the sensing capability of 
 ${\bf X}_0$ 
while simultaneously supporting communications.

It is  observed that Theorem~\ref{Theorem2} is validated: 
in the bistatic ISAC system, 
 the DRIS does not always degrade the sensing performance.
In contrast, the DRIS only disrupts the  estimation accuracy of the AoD $\theta_1$, 
while simultaneously enhancing the estimation accuracy of the AoA $\theta_2$.
% This anisotropic impact induced by the DRIS on the estimation accuracy of different angular parameters can be exploited. 
% For example, 
Note that the anisotropic impact induced by the DRIS  can be used as a controllable resource to improve AoA estimation in scenarios where precise localization or tracking relative to the sensing receiver is essential. 
Conversely, the degradation in AoD accuracy can be intentionally used to mask transmitter-related angular information.  
%This offers an additional layer of protection and interference resilience for the sensing performance in bistatic ISAC systems. 
%Hence, the DRIS serves as a dual-purpose mechanism that enhances performance in certain angular domains while enabling intentional concealment in others.

It is worth noting that, in Fig.~\ref{ResFigSensingvsPowerMLE}, the AoA $\theta_2$ estimates are more accurate than the AoD $\theta_1$ estimates, even without a DRIS.
Under the considered geometry in this section, the sensing receiver observes the target only within a narrow angular sector around its broadside direction. 
In contrast, the ISAC BS observes the target at more oblique departure angles. 
Since the steering-vector derivative with respect to the AoA is therefore stronger on average, 
the corresponding Fisher information is larger. 
This leads to a lower CRLB and smaller empirical MSE than for the AoD.
 
\begin{figure}[!t]
    \centering
    \includegraphics[scale=0.63]{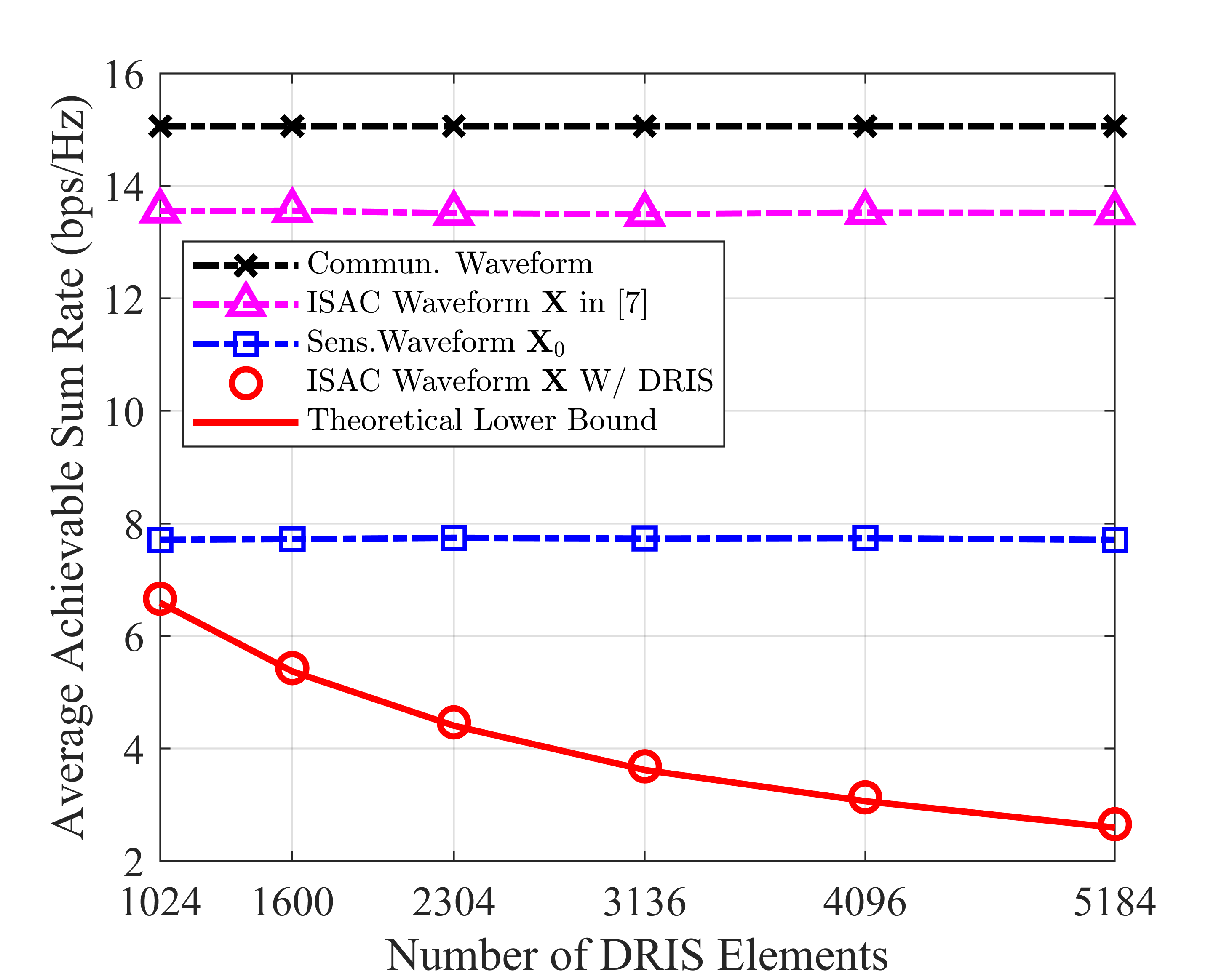}
    \caption{Average achievable sum rate vs. number of DRIS elements for different benchmarks.}
    \label{ResfigCND}
\end{figure}

\begin{figure*}[!t]
    \centering
    \subfloat{
            \includegraphics[scale=0.62]{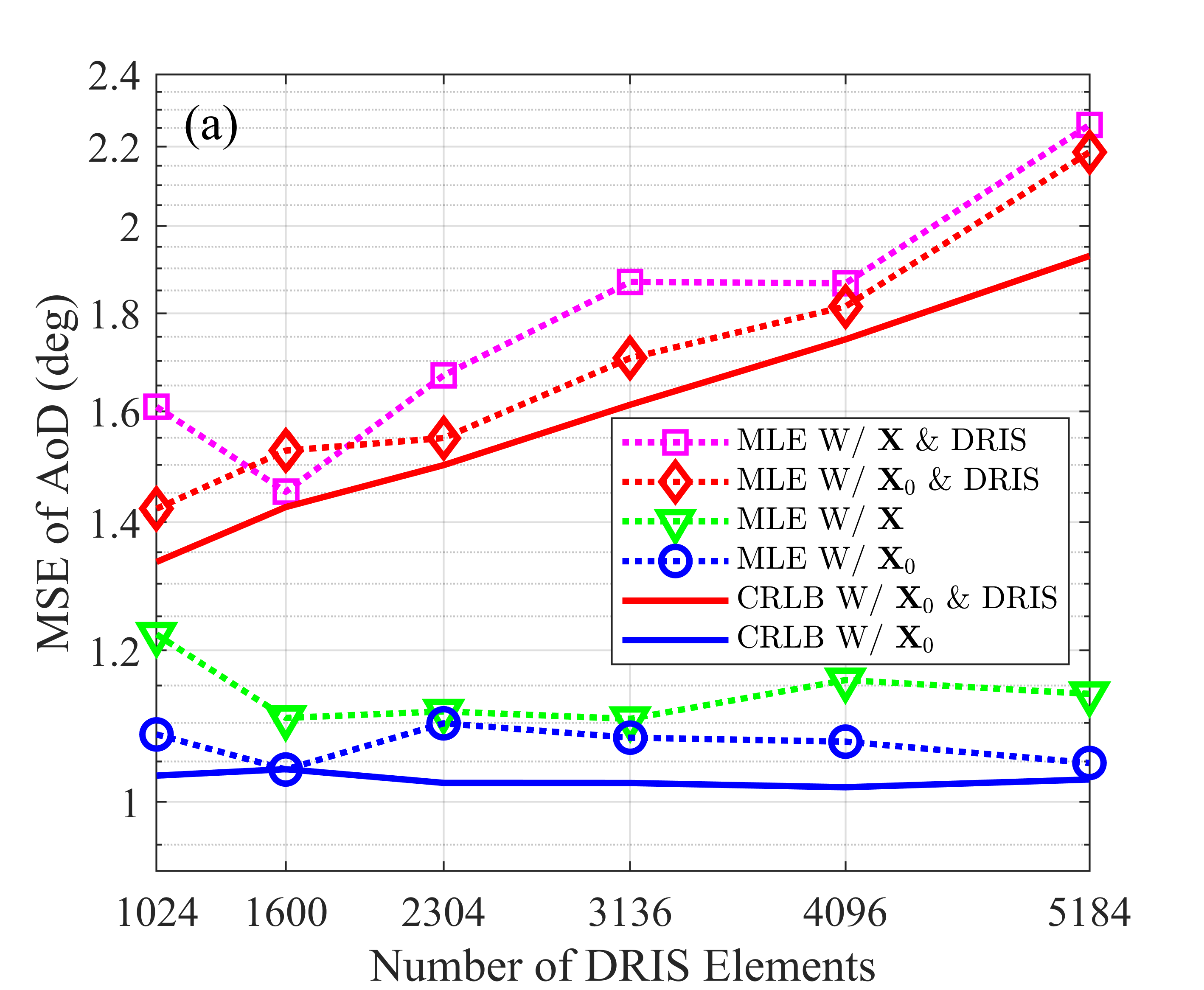}}\hspace{25pt}
    \subfloat{
            \includegraphics[scale=0.62]{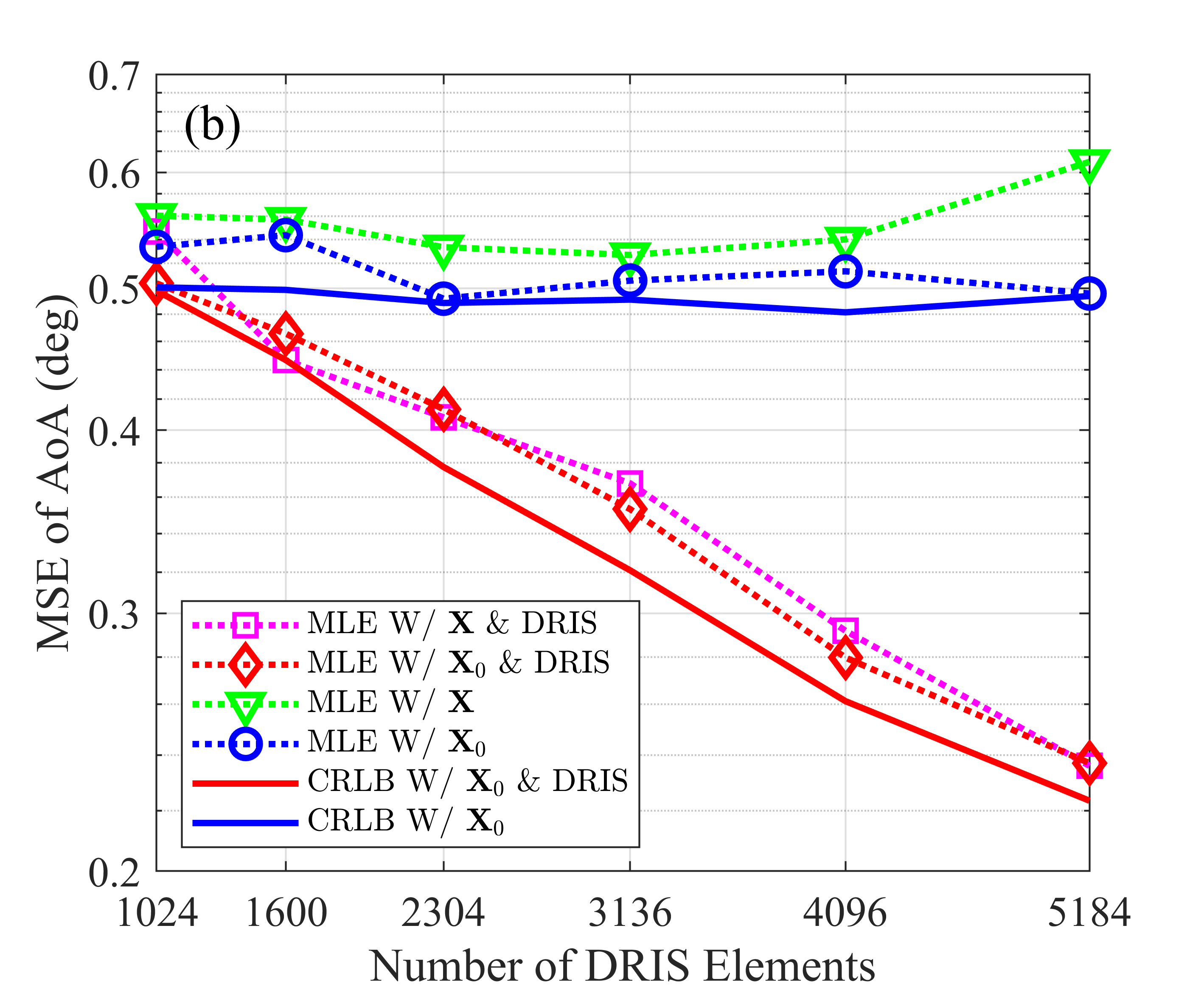}}
   \caption{Mean Square Error (MSE) of (a)  AoD ($\theta_1$) and (b) AoA ($\theta_2$) estimates versus 
   number of DRIS elements
for different ISAC waveforms under MLE.}
    \label{ResFigSensingvsNDMLE}
\end{figure*} 
Fig.~\ref{ResfigCND} illustrates the relationship between the average achievable sum rate 
and the number of DRIS elements when the total transmit power is set at 11 dBm.
As expected, the conventional communication-only system achieves the highest achievable sum rate 
since it is optimized solely for data transmission and is free from 
the sensing or DRIS constraint.
The ISAC waveform $\bf X$ attains a slightly lower but still high sum rate, 
reflecting the intrinsic tradeoff between sensing and communication performance. 
In contrast, the optimal sensing waveform ${\bf X}_0$
yields the lowest sum rate among these benchmarks
(Commun. Waveform and the ISAC Waveform $\bf X$ in~\cite{FLiuISAC}), 
because it strictly enforces the covariance constraint in~\eqref{ISACWaveLim}
to maximize sensing performance, which inevitably reduces the achievable communication rate.

When the DRIS is deployed, the behavior changes dramatically. 
According to Theorem~\ref{Theorem1}, 
the impact of the DRIS on communication performance increases with 
the number of DRIS elements $N_{\rm D}$.
Specifically, the lower bound of the received SINR contains an ACA term proportional to
 ${ P_0 {{{\mathscr{L}}^{\rm c}_{{\rm {cas}},k}} {N\!_{\rm D}}{\overline \mu} }}$ in the denominator, 
 so that increasing the number of DRIS elements linearly amplifies the ACA impact.
As a result, the average sum rate 
achieved by  ISAC Waveform $\bf X$ W\ DRIS monotonically decreases with $N_{\rm D}$.
The theoretical lower bound predicted by Theorem~\ref{Theorem1}
closely tracks the simulated curve over $N_{\rm D}$,
confirming the accuracy of the derived bound 
in the presence of a DRIS.
Therefore, Theorem~\ref{Theorem1} provides an analytically tractable characterization of this
DRIS impact.

 Fig.~\ref{ResFigSensingvsNDMLE} illustrates the
MSE results and the CRLBs of the (a) AoD $\theta_1$ and (b) AoA $\theta_2$ estimates 
as the number of DRIS elements $N_{\rm D}$ increases for different benchmarks.
For the AoD estimation results plotted in Fig.~\ref{ResFigSensingvsNDMLE}~(a),
the CRLB curve corresponding to 
CRLB W/ ${\bf X}_0$
 is essentially flat and close to $1.02^\circ$, 
 while the CRLB in the presence of a DRIS, i.e., 
 CRLB W/ ${\bf X}_0$ \& DRIS grows monotonically from approximately 
$1.33^\circ$ to over $1.93^\circ$ as the number of DRIS elements $N_{\rm D}$ increases.
One can see that the MSE achieved by the MLE closely tracks the two CRLBs, 
with only minor fluctuations due to Monte-Carlo averaging.

Based on the property derived from Theorem~\ref{Theorem2},
the introduction of the DRIS disrupts the estimation
accuracy of the AoD $\theta_1$ while enhancing the estimation accuracy
of the AoA $\theta_2$.
The behavior in Fig.~\ref{ResFigSensingvsNDMLE}~(b) for the AoA differs markedly, 
which confirms the anisotropic property predicted by Theorem~\ref{Theorem2}.
In the absence of a DRIS, the CRLB for $\theta_2$ is approximately $0.50^\circ$.
Meanwhile, the corresponding MSEs obtained with ${\bf X}_0$ and ${\bf X}$
lie slightly above the CRLB. 
In the presence of a DRIS, however, both the CRLB and the MSE results achieved by MLE decrease as $N_{\rm D}$ grows.
Specifically, the CRLB drops from around $0.50^\circ$ to around $0.22^\circ$.

Fig.~\ref{ResfigCDis} illustrates the average achievable sum rate versus the distance between the ISAC BS and the DRIS,
with the DRIS located at (-$d_{\rm D}$ m, 0 m, 2.5 m) and $d_{\rm D} = 0.5, 1, \ldots, 3$. 
For ISAC Waveform $\bf X$ W/ DRIS and the corresponding theoretical bound,
the degradation trend is consistent with the SINR lower bound 
derived in Theorem~\ref{Theorem1}.
The denominator of SINR in~\eqref{SINRDRISMath} 
contains ${ P_0 {{{\mathscr{L}}^{\rm c}_{{\rm {cas}},k}} {N\!_{\rm D}}{\overline \mu} }}$,
where the cascaded large-scale channel fading coefficient ${{\mathscr{L}}^{\rm c}_{{\rm {cas}},k}}$
depends on the distance between the ISAC BS and the DRIS.
\begin{figure}[!t]
    \centering
    \includegraphics[scale=0.63]{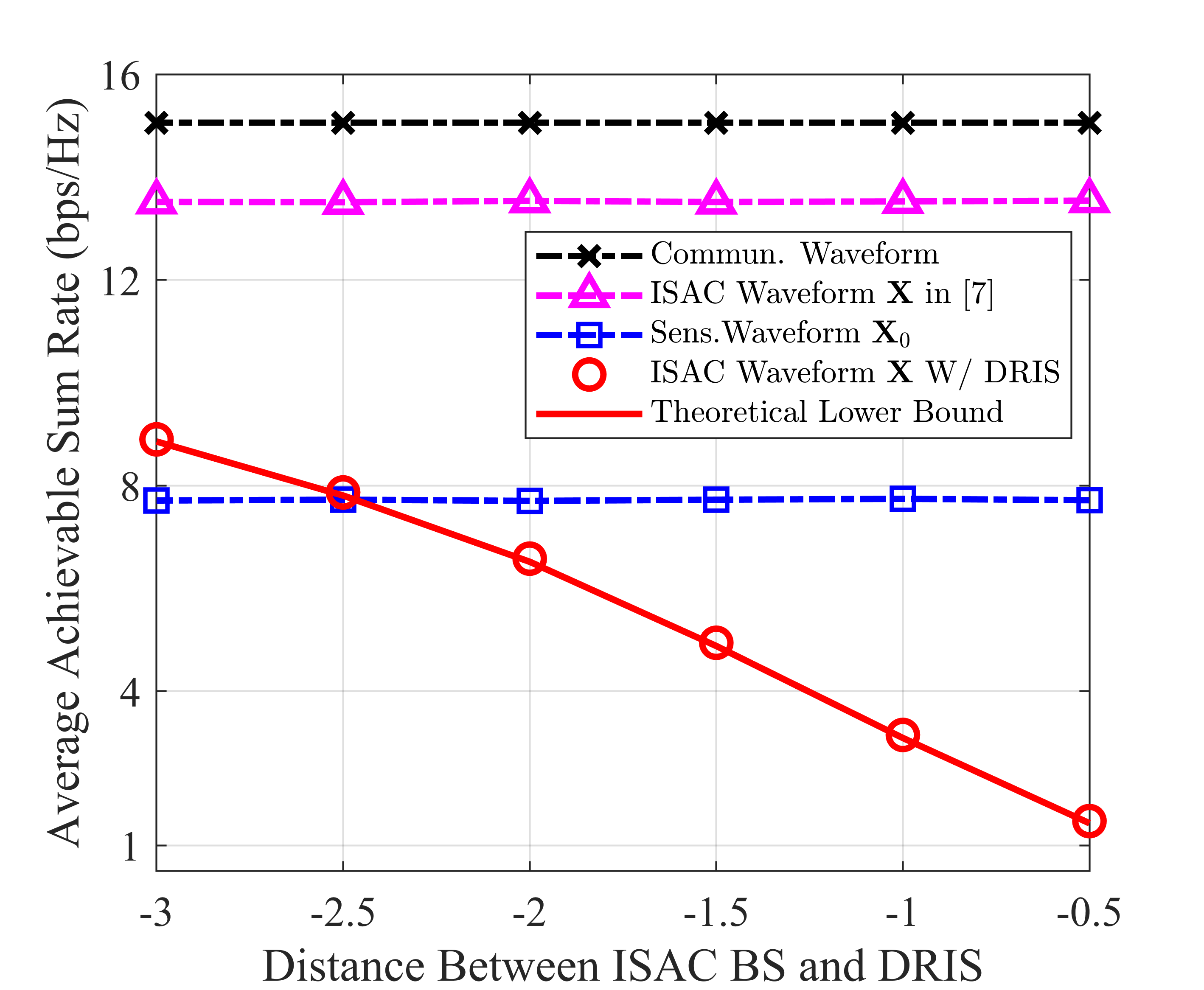}
    \caption{Average achievable sum rate vs. distance between ISAC BS and DRIS for different benchmarks.}
    \label{ResfigCDis}
\end{figure}

Increasing the distance between the ISAC BS and the DRIS results in a larger
${{\mathscr{L}}^{\rm c}_{{\rm {cas}},k}}$ and thus reduces the impact of 
DRIS-induced ACA interference on the achievable sum rate.
Deploying the DRIS close to the ISAC BS may increase the risk of being detected. 
Fortunately, an interesting property of RISs is their passive nature,
where RISs cannot transmit/receive, or process signals~\cite{IRSCuiTJ}. 
In addition, an RIS can be easily hidden in the environment by disguising it, 
for instance, by embedding it in a glass structure. 
In practical scenarios, an RIS can also be mounted on walls or integrated into existing infrastructure. 
Therefore, our DRIS can be set up ahead of time and hidden using the aforementioned camouflage techniques.

In Fig.~\ref{ResFigSensingvsDisMLE}, we respectively show the
MSE results and the CRLBs of the (a) AoD $\theta_1$ and (b) AoA $\theta_2$ estimates 
versus the distance between the ISAC BS and the DRIS for different benchmarks.
For the estimation results of AoD in Fig.~\ref{ResFigSensingvsDisMLE} (a),
the simulation results of CRLB W/O DRIS stay close to $1^\circ$.
Meanwhile, the corresponding MLE curves, i.e., MLE W/ ${\bf X}$
and MLE W/ ${\bf X}_0$ remain slightly above the CRLB.
When the DRIS is implemented, however, 
the simulation results of CRLB W/ DRIS 
increase significantly as the DRIS approaches the ISAC BS.
Specifically, the CRLB changes from $1.19^\circ$ to $3.05^\circ$
as the coordinate varies from $-3$ m to $-0.5$ m.  
The MLE results of AoD estimates, i.e., MLE W/ ${\bf X}$ \& DRIS and MLE W/ ${\bf X}_0$ \& DRIS
also remain close to the CRLB, thus verifying the  derived Theorem~\ref{Theorem2}.
According to~\eqref{CovRDR1}, ${\widetilde {\bf {R}}}_l$ 
does not depend on the AoD $\theta_1$, while 
the DRIS contributes additional ACA interference.
As a result, the DRIS reduces the FIM for $\theta_1$
and enlarges the corresponding CRLB as the distance between the ISAC BS and the DRIS decreases.
\begin{figure*}[!t]
    \centering
    \subfloat{
            \includegraphics[scale=0.62]{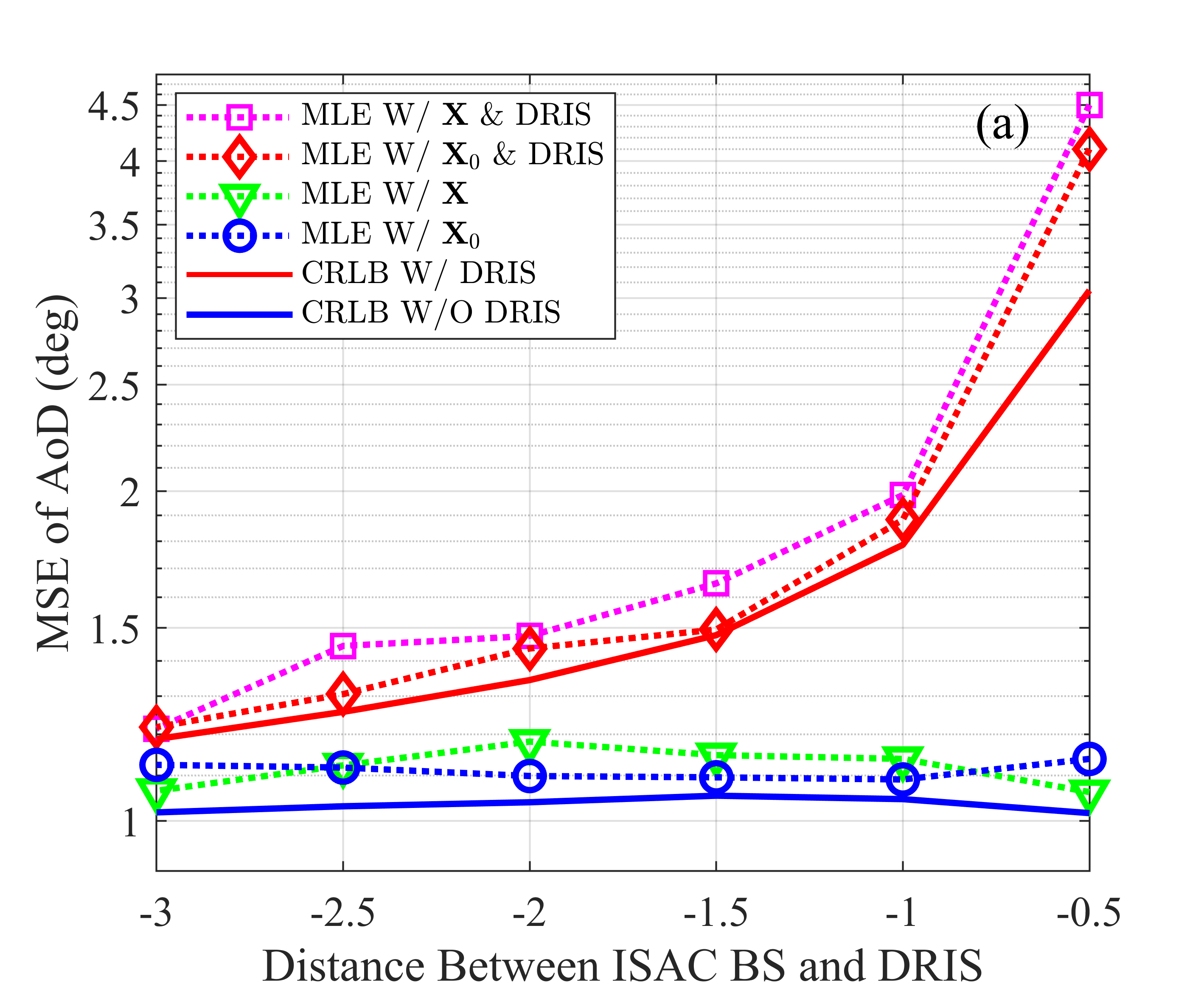}}\hspace{25pt}
    \subfloat{
            \includegraphics[scale=0.62]{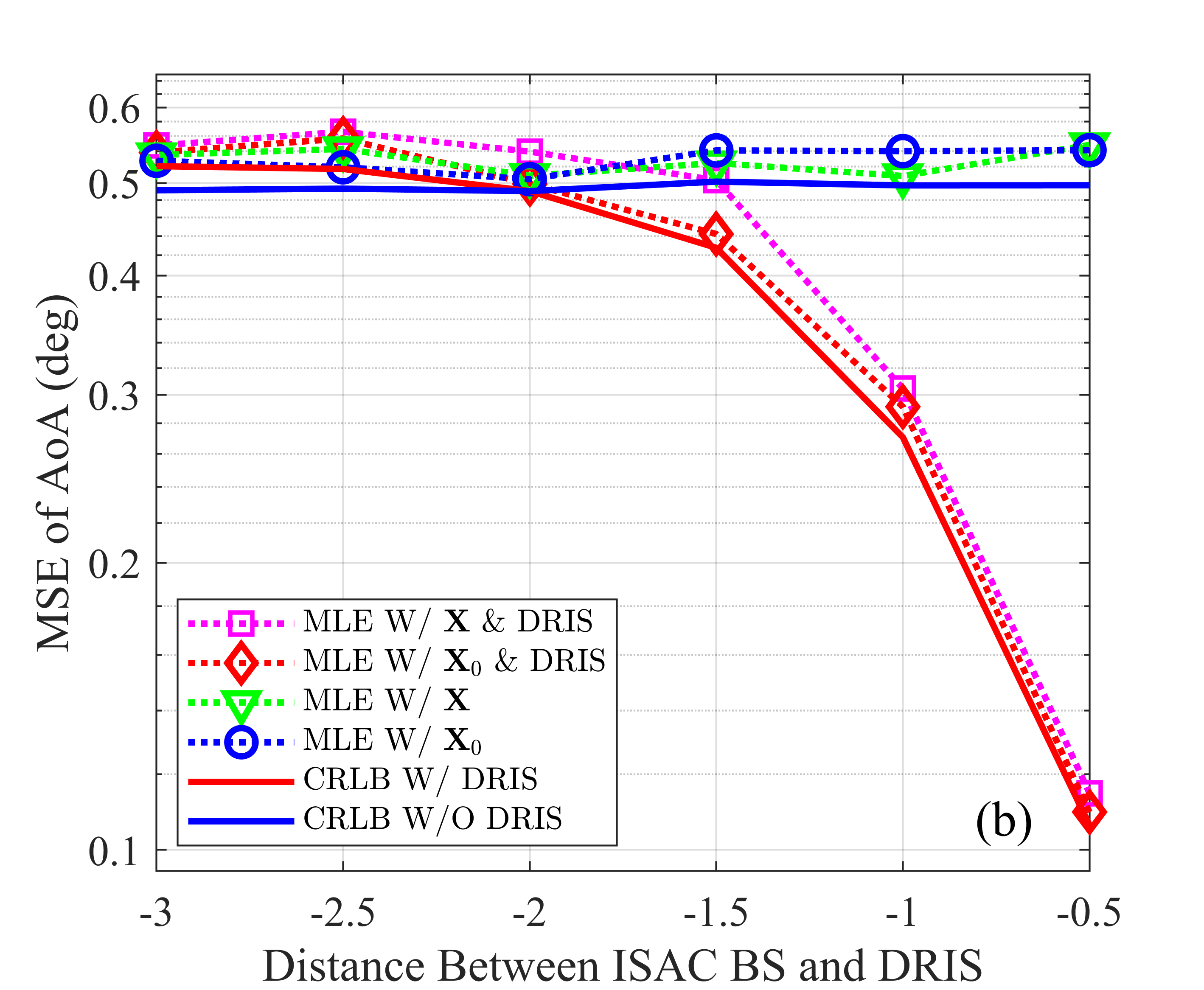}}
   \caption{Mean Square Error (MSE) of (a)  AoD ($\theta_1$) and (b) AoA ($\theta_2$) estimates versus 
   distance between ISAC BS and DRIS
for different ISAC waveforms under MLE.}
    \label{ResFigSensingvsDisMLE}
\end{figure*} 

Fig.~\ref{ResFigSensingvsDisMLE}~(b) plots the AoA estimation results. 
One can see that the MSE results and the CRLBs of AoA $\theta_2$ estimates differ markedly from those of the  AoD.
In the presence of a DRIS,
both the CRLB and the MLE curves decrease rapidly as the DRIS is deployed closer to the ISAC BS.
The CRLB changes from $0.52^\circ$ to $0.106^\circ$
as the coordinate varies from $-3$ m to $-0.5$ m.
This improvement of estimation accuracy arises from the explicit dependence of ${\widetilde {\bf {R}}}_l$ 
on $\theta_2$ via the ${ \boldsymbol{\alpha}}_{{\!N\!_{\rm S}}} \!\! \left({\theta_2}\right)$ in~\eqref{CovRDR1}.
As seen in Theorem~\ref{Theorem2}, 
the FIM involving $\frac{ {\partial { {\widetilde {\bf {R}}} _{l}}}}{{\partial {\theta _2}}}$
scales with both the cascaded large-scale gain  and the AoA sensitivity.
Therefore, less propagation loss in the cascaded DRIS-based channel 
 increases the Fisher information about  the AoA $\theta_2$.
From Fig.~\ref{ResFigSensingvsDisMLE}~(b),
the close match between the MLE curves and the CRLBs confirms the derivations in~\eqref{FIM}.

\section{Conclusions}\label{Conclu}
In this work, we have  investigated the impact of a DRIS on bistatic ISAC systems.
To the best of our knowledge, this is the first time that the impact of a DRIS on bistatic ISAC systems has been quantitatively characterized.
Specifically, we have derived a theoretical lower bound of the SINRs and closed-form CRLBs for the AoD and AoA estimates 
to characterize the communication and sensing performance in the presence of a DRIS, respectively.
Our theoretical analysis and numerical results lead to the following conclusions,
highlighting the unique properties of the DRIS for bistatic ISAC systems.
\begin{itemize}
    \item  The implementation of the DRIS in the bistatic ISAC system introduces additional 
time-varying ACA channels for communication users. 
As a result, the communication performance is degraded due to the DRIS-induced ACA channels.
Mathematically, we have derived the statistical characteristics of the ACA channels.
Specifically, the elements of the ACA channels for the $k$-th communication user 
converge in distribution to a complex Gaussian distribution
with zero mean and variance of  ${{{\mathscr{L}}^{\rm c}_{{\rm {cas}},k}} {N\!_{\rm D}}{\overline \mu} }$,
as the number of DRIS reflective elements becomes large.
\item Regarding the bistatic sensing performance,
the use of DRIS introduces an additional time-varying DRIS-based sensing path for the target.
We have derived the statistical characteristics of this time-varying path. 
Mathematically, its elements converge in distribution to a complex Gaussian distribution
with zero mean and variance of  ${{{\mathscr{L}}^{\rm s}_{{\rm {cas}}}} {N\!_{\rm D}}{\overline \nu} }$,
as the number of DRIS reflective elements becomes large.
It is worth noting that the DRIS only decreases the estimation accuracy of the AoD while improving that of the AoA.
\item   Based on our theoretical derivations,
the lower bound of the received SINR contains an ACA term proportional to
 ${ P_0 {{{\mathscr{L}}^{\rm c}_{{\rm {cas}},k}} {N\!_{\rm D}}{\overline \mu} }}$ in the denominator.
 It can be seen that, unlike conventional AJs, 
 increasing the transmit power at the ISAC BS not only fails to mitigate the DRIS-based ACA interference, 
 but in fact amplifies its impact.
 On the other hand, 
the introduction of DRIS transforms the covariance matrix ${\widetilde {\bf {R}}}_l$
from ${\delta^2_{\rm s}}{\bf I}\!_{N_{\rm S}}$ to $\chi^2 \!{{\mathscr{L}}^{\rm s}_{\!{\rm {d}},2}}{{\mathscr{L}}^{\rm s}_{\!{\rm {cas}}}} {N\!_{\rm D}}{\overline \nu} 
    {\left\|{\boldsymbol{x}}_{l}\right\|^2}\!\!
    \left(  { \boldsymbol{\alpha}}^H_{{\!N\!_{\rm S}}} \!\!\left({\theta_2}\right) \!{ \boldsymbol{\alpha}}_{{\!N\!_{\rm S}}} \!\! \left({\theta_2}\right) \! \right) + {\delta^2_{\rm s}}{\bf I}\!_{N_{\rm S}}$
compared to the bistatic ISAC system without DRIS.
As a result, the derived closed-form CRLBs of AoA and AoD estimates clearly show that the DRIS exhibits an anisotropic impact:
 it only degrades the estimation accuracy of the AoD while simultaneously improving that of the AoA.
\end{itemize}

\begin{appendices}
\section{Proof of Proposition~\ref{Proposition1}}\label{AppendixA}
According to the definition of ${\bf H}^{\rm c}_{\!{\rm{A\!C\!A}}}$ in~\eqref{ACAEq},
the channel element $\left[{\bf H}^{\rm c}_{\!{\rm{A\!C\!A}}}\right]_{n,k}$
can be further expressed as  
    \begin{alignat}{1}
        \nonumber
        &{\left[ {{{\bf{H}}^{\rm c}_{\!{\rm{A\!C\!A}}}}} \right]_{n,k}} = \\ \nonumber
        & \sqrt {\frac{{{\varepsilon}{{{\mathscr{L}}}_{\rm{G}}}
        {{\mathscr{L}}^{\rm c}_{{\rm{I}},k}}}}{{{\varepsilon} + 1}}}
        \left[{\widehat {\bf{G}}}^{{\rm{LOS}}}\right]_{n,:} \!
        {\rm{diag}}\! \left({\boldsymbol{\varphi}}(t_{\rm{\!D\!T}}) - {\boldsymbol{\varphi}}(t_{\rm{\!P\!T}})\right)
        {\widehat {{\boldsymbol{h}}}^{\rm c}_{{\rm I},k}}  \\
        &+ \sqrt {\frac{{{{{\mathscr{L}}}_{\rm{G}}}
        {{\mathscr{L}}^{\rm c}_{{\rm{I}},k}}}}{{{\varepsilon} + 1}}}
        \left[{\widehat {\bf{G}}}^{{\rm{NLOS}}}\right]_{n,:} \!
        {\rm{diag}}\! \left({\boldsymbol{\varphi}}(t_{\rm{\!D\!T}}) - {\boldsymbol{\varphi}}(t_{\rm{\!P\!T}})\right)
        {\widehat {{\boldsymbol{h}}}^{\rm c}_{{\rm I},k}},
        \label{RewriHACAele}
    \end{alignat}
where $n=1,\ldots, N_{\rm B}$ and $k=1,\ldots , K_{\rm c}$.
Consequently, we define the following terms:
\begin{alignat}{1}
        {\left[ {{{\widehat{\bf{H}}}^{\rm c}_{\!{\rm{A\!C\!A}}}}} \right]^{\rm{\!LOS}}_{n,k}} &= 
        \left[{\widehat {\bf{G}}}^{{\rm{LOS}}}\right]_{n,:} 
        \!\odot \! \left({\boldsymbol{\varphi}}(t_{\rm{\!D\!T}}) - {\boldsymbol{\varphi}}(t_{\rm{\!P\!T}})\right)
        {\widehat {{\boldsymbol{h}}}^{\rm c}_{{\rm I},k}},  \label{HACAeleLOSLOS} \\
        {\left[ {{{\widehat{\bf{H}}}^{\rm c}_{\!{\rm{A\!C\!A}}}}} \right]^{\rm{\!NLOS}}_{n,k}}  & = 
        \left[{\widehat {\bf{G}}}^{{\rm{NLOS}}}\right]_{n,:} \!
        \!\odot \! \left({\boldsymbol{\varphi}}(t_{\rm{\!D\!T}}) - {\boldsymbol{\varphi}}(t_{\rm{\!P\!T}})\right)
        {\widehat {{\boldsymbol{h}}}^{\rm c}_{{\rm I},k}}.
        \label{HACAeleLOSNLOS}
    \end{alignat}
Furthermore,  
${\left[ {{{\widehat{\bf{H}}}^{\rm c}_{\!{\rm{A\!C\!A}}}}} \right]^{\rm{\!LOS}}_{n,k}}$ and
${\left[ {{{\widehat{\bf{H}}}^{\rm c}_{\!{\rm{A\!C\!A}}}}} \right]^{\rm{\!NLOS}}_{n,k}}$ 
are written as
\begin{alignat}{1}
{\left[ {{{\widehat{\bf{H}}}^{\rm c}_{\!{\rm{A\!C\!A}}}}} \right]^{\! \rm LOS}_{n,k}}
& = \sum\limits_{r = 1}^{N_{\rm D}} {h^{\! \rm LOS}_{r}}
\label{RewriHACAeleLOSsim} \\
\nonumber
&= \sum\limits_{r = 1}^{N_{\rm D}} \!\Big(
\! \left[{\widehat {\bf{G}}}^{\rm LOS}\right]_{n,r}\!\!\big( {\beta _r}(t_{\rm{D\!T}}) e^{j{\varphi_r}(t_{\rm{D\!T}})} \\
&\;\;\;\;\;\;\;\;\;\;\;\;  - {\beta _r}(t_{\rm{P\!T}}) e^{j{\varphi_r}(t_{\rm{P\!T}})} \big) 
\left[{\widehat {{\boldsymbol{h}}}^{\rm c}_{{\rm I},k}}\right]_{r}  \!\Big),
\label{RewriHACAeleLOS}  
\end{alignat}
and
\begin{alignat}{1}
{\left[ {{{\widehat{\bf{H}}}^{\rm c}_{\!{\rm{A\!C\!A}}}}} \right]^{\! \rm NLOS}_{n,k}}
& = \sum\limits_{r = 1}^{N_{\rm D}} {h^{\! \rm NLOS}_{r}}
\label{RewriHACAeleNLOSsim} \\
\nonumber
&= \sum\limits_{r = 1}^{N_{\rm D}} \!\Big( 
\! \left[{\widehat {\bf{G}}}^{\rm NLOS}\right]_{n,r}\!\!\big( {\beta _r}(t_{\rm{D\!T}}) e^{j{\varphi_r}(t_{\rm{D\!T}})} \\
&\;\;\;\;\;\;\;\;\;\;\;\;\;   - {\beta _r}(t_{\rm{P\!T}}) e^{j{\varphi_r}(t_{\rm{P\!T}})} \big) 
\left[{\widehat {{\boldsymbol{h}}}^{\rm c}_{{\rm I},k}}\right]_{r}  \!\Big).
\label{RewriHACAeleNLOS} 
\end{alignat}

Conditioned on the fact that the random variables in $\bf G$, ${{\boldsymbol{\varphi}}(t)}$, 
and ${\widehat {{\boldsymbol{h}}}^{\rm c}_{{\rm I},k}}$ are independent, 
we have %${\mathbb{E}}\!\left[ {h^{\! \rm LOS}_{r}} \right] = {\mathbb{E}}\!\left[ {h^{\! \rm NLOS}_{r}} \right] = 0$.
% Namely, 
\begin{alignat}{1}
        {\mathbb{E}}\!\!\left[ \!{\left[ {{{\widehat{\bf{H}}}^{\rm c}_{\!{\rm{A\!C\!A}}}}} \right]^{\! \rm LOS}_{n,k}} \right] & = 0,  \label{exHACAeleLOSLOS} \\
        {\mathbb{E}}\!\!\left[ \!{\left[ {{{\widehat{\bf{H}}}^{\rm c}_{\!{\rm{A\!C\!A}}}}} \right]^{\! \rm NLOS}_{n,k}} \right]  & = 0. \label{exHACAeleLOSNLOS}
\end{alignat}

On the other hand, the variance of ${h^{\! \rm LOS}_{r}}$ is calculated as
            % {\mathbb{E}}\!\left[ {{h^{\! \rm LOS}_{r}}} \left({{h^{\! \rm LOS}_{r}}}\right)^H  \right] \\
            % \nonumber
    \begin{alignat}{1}
            \nonumber
             {\mathbb{D}} \!\!\left[ {h^{\! \rm LOS}_{r}} \right]  
            &= {\mathbb{E}}  \!\big[\!\left|{\beta _r}(t_{\rm{D\!T}})\right|^2 \!+\! 
            \left|{\beta _r}(t_{\rm{P\!T}})\right|^2 \\
            & \; -2{\beta _r}(t_{\rm{D\!T}}){\beta _r}(t_{\rm{P\!T}})\cos\!\left({\varphi_r}(t_{\rm{D\!T}}) \!-\! {\varphi_r}(t_{\rm{P\!T}})\right)\big].
        \label{VarLOS}
    \end{alignat} 
More specifically, $ {\mathbb{D}} \!\!\left[ {h^{\! \rm LOS}_{r}} \right]$ can be represented by
\begin{alignat}{1}
    &{\mathbb{D}} \!\!\left[ {h^{\! \rm LOS}_{r}} \right] = \overline \mu \\
    &= \sum\limits_{{i_1} = 1}^{{2^b}} {\sum\limits_{{i_2} = 1}^{{2^b}} {{p_{{i_1}}}{p_{{i_2}}}
    \!\left(\!{\mu _{{i_1}}^2 \!+\! \mu _{{i_2}}^2 \!-\! 2{\mu _{{i_1}}}{\mu _{{i_2}}}\cos\! 
    \left(\! {{\phi _{{i_1}}} - {\phi _{{i_2}}}} \right)} \!\right)} } .
    \label{VarLOSf}
\end{alignat}
Similarly, the variance of ${h^{\! \rm NLOS}_{r}}$ is computed as
\begin{equation} 
    {\mathbb{D}}\! \!\left[ {h^{\! \rm NLOS}_{r}} \right] = \overline \mu .
    \label{VarNLOSf}
\end{equation} 

Therefore, the variances of ${\left[ {{{\widehat{\bf{H}}}^{\rm c}_{\!{\rm{A\!C\!A}}}}} \right]^{\! \rm LOS}_{n,k}}$
and 
${\left[ {{{\widehat{\bf{H}}}^{\rm c}_{\!{\rm{A\!C\!A}}}}} \right]^{\! \rm NLOS}_{n,k}}$ can be respectively given by
\begin{alignat}{1}
        {\mathbb{D}} \!\!\left[ \!{\left[ {{{\widehat{\bf{H}}}^{\rm c}_{\!{\rm{A\!C\!A}}}}} \right]^{\! \rm LOS}_{n,k}} \right] & ={\frac{{{\varepsilon}{{{\mathscr{L}}}_{\rm{G}}}
        {{\mathscr{L}}^{\rm c}_{{\rm{I}},k}}}N_{\! \rm D} {\overline \mu} }{{{\varepsilon} + 1}}},  \label{VarHACAeleLOSLOS} \\
        {\mathbb{D}} \!\!\left[ \!{\left[ {{{\widehat{\bf{H}}}^{\rm c}_{\!{\rm{A\!C\!A}}}}} \right]^{\! \rm NLOS}_{n,k}} \right]  & = {\frac{{ {{{\mathscr{L}}}_{\rm{G}}}
        {{\mathscr{L}}^{\rm c}_{{\rm{I}},k}}}N_{\! \rm D} {\overline \mu} }{{{\varepsilon} + 1}}}. \label{VarexHACAeleLOSNLOS}
\end{alignat}
Based on the Lindeberg-L$\acute{e}$vy central limit theorem, we have
\begin{equation} 
    \frac{{\left[ {{{\bf{H}}^{\rm c}_{\!{\rm{A\!C\!A}}}}} \right]_{n,k}}}
    {{\sqrt {N_{\rm D}{{{\mathscr{L}}}_{\rm{G}}}
        {{\mathscr{L}}^{\rm c}_{{\rm{I}},k}}{\overline \mu}}} }  
    \mathop  \to \limits^{\rm{d}} \mathcal{CN}\left( {0,{1}} \right),
    \label{HACALTeq} 
\end{equation}
as ${N_{\rm D}} \to \infty$. 

As a result, Proposition~\ref{Proposition1} can be proved.

\section{Proof of Proposition~\ref{Proposition2}}\label{AppendixB}
Based on the expression of the time-varying DRIS-based sensing path ${\boldsymbol{h}_{\rm D}^{\rm s}}\!\left(t\right)$,
we can rewrite the $n$-th element of ${\boldsymbol{h}_{\rm D}^{\rm s}}\!\left(t\right)$ as the following form:
\begin{alignat}{1}
    \nonumber
     {{h}_{{\rm D},n}^{\rm s}}\!(t)  & =\left[ {\bf G}{\rm{diag}}\!\left({{\boldsymbol{\varphi}}(t)}\right){{{ \boldsymbol{h}}}^{\rm s}_{{\rm I}}}\right]_{n} =\\ \nonumber
    & \sqrt {\frac{{{\varepsilon}{{{\mathscr{L}}}^{\rm s}_{{\rm{cas}},k}}  }}{{{\varepsilon} + 1}}}\left[{\widehat {\bf{G}}}^{{\rm{LOS}}}\right]_{n,:}
     \!  {\rm diag}\!\left({{\boldsymbol{\varphi}}(t)}\right)   {\widehat {{ \boldsymbol{h}}}^{\rm s}_{{\rm I}}} \\
    & + \sqrt {\frac{{{{{\mathscr{L}}}^{\rm s}_{{\rm{cas}},k} } }}{{{\varepsilon} + 1}}}\left[{\widehat {\bf{G}}}^{{\rm{NLOS}}}\right]_{n,:}
    \!{\rm diag}\!\left({{\boldsymbol{\varphi}}(t)}\right)  {\widehat {{ \boldsymbol{h}}}^{\rm s}_{{\rm I}}}.
    \label{hDsRewir}
\end{alignat}
Therefore, \eqref{hDsRewir} can be written as 
\begin{alignat}{1}
    \nonumber
    {{h}_{{\rm D},n}^{\rm s}}\!(t) &=  \sqrt {\!\frac{{\varepsilon}{{\mathscr{L}}^{\rm s}_{{\rm{cas}}}} \!N\!_{\rm D}}{{1\!+\!{\varepsilon}} }}
     \!\!\left( \!{\frac{\sum\limits_{r = 1}^{{N_{\rm{D}}}} {\!{ {{e^{ \!- \!j\!\frac{2\pi}{\lambda}\left(\!{D^{r}_{n}} -{D_n}\right)}}}} {\beta_{r}\!(t)}{e^{j{\varphi_{r}}(t)}} {{{\widehat {  {h} }}^{\rm s}_{{\rm{I}},r}}} } }{\sqrt{\!N_{\rm D}}}} \!\right) \\
    &+ \sqrt {\!\frac{ {{\mathscr{L}}^{\rm s}_{{\rm{cas}}}} \!N\!_{\rm D}}{{1\!+\!{\varepsilon}} }}
    \!\!\left( \!{\frac{\sum\limits_{r = 1}^{{N_{\rm{D}}}} {\!{ \left[{\widehat {\bf{G}}}^{{\rm{NLOS}}}\right]_{n,r}} {\beta_{r}\!(t)}{e^{j{\varphi_{r}}(t)}} {{{\widehat {  {h} }}^{\rm s}_{{\rm{I}},r}}} } }{\sqrt{\!N_{\rm D}}}} \!\right),
    \label{hDsRewirextend}
\end{alignat}
where ${{{\widehat {  {h} }}^{\rm s}_{{\rm{I}},r}}}$ is the $r$-th element 
of ${\widehat {{ \boldsymbol{h}}}^{\rm s}_{{\rm I}}}$. 
According to the definition in~\eqref{hIs}, ${{{\widehat {  {h} }}^{\rm s}_{{\rm{I}},r}}}$ can 
be specifically expressed by
\begin{equation}
    \nonumber
    {{{\widehat {  {h} }}^{\rm s}_{{\rm{I}},r}}} = e^{j2\pi\Delta \left(\!\left\lfloor {\frac{r}{{{N_{ {\rm{D}},h}}}}} \right\rfloor -1\!\right)\sin\!{\vartheta\!_h}}
    e^{j2\pi\Delta \left(\!r-{{N_{ {\rm{D}},h}}}\left\lfloor {\frac{r}{{{N_{ {\rm{D}},h}}}}} \right\rfloor -1\!\right)\sin\!{\vartheta\!_v}},
\end{equation}
where $\left\lfloor x \right\rfloor$ denotes the largest integer less than or equal to $x$.

For analytical convenience, we define the two equations as follows:
\begin{alignat}{1}
    a_r &=  {\frac{  {\!{ {{e^{ \!- \!j\!\frac{2\pi}{\lambda}\left(\!{D^{r}_{n}} -{D_n}\right)}}}} {\beta_{r}\!(t)}{e^{j{\varphi_{r}}(t)}} {{{\widehat {  {h} }}^{\rm s}_{{\rm{I}},r}}} } }{\sqrt{\!N_{\rm D}}}},  \label{hDsRewirextenda}\\
    b_r &= {\frac{  {\!{ \left[{\widehat {\bf{G}}}^{{\rm{NLOS}}}\right]_{n,r}} {\beta_{r}\!(t)}{e^{j{\varphi_{r}}(t)}} {{{\widehat {  {h} }}^{\rm s}_{{\rm{I}},r}}} } }{\sqrt{\!N_{\rm D}}}}.
    \label{hDsRewirextendb}
\end{alignat}

Given that the random elements in $\bf G$, ${{\boldsymbol{\varphi}}(t)}$, and ${{{ \boldsymbol{h}}}^{\rm s}_{{\rm I}}}$ are independent,
the expectation of ${{h}_{{\rm D},n}^{\rm s}}\!(t)$ is then calculated as ${\mathbb{E}}\!\left[{{h}_{{\rm D},n}^{\rm s}}\!(t)\right] = 0$.
Furthermore, the variance of ${{h}_{{\rm D},n}^{\rm s}}\!(t)$ is reduced to
\begin{equation}
   {\mathbb{D}}\!\!\left[{{h}_{{\rm D},n}^{\rm s}}\!(t)\!\right] \!\!=
   {\!\frac{{\varepsilon}{{\mathscr{L}}^{\rm s}_{{\rm{cas}}}} \!N\!_{\rm D}}{{1\!+\!{\varepsilon}} }} {\sum\limits_{r = 1}^{{N_{\rm{D}}}} 
  \! {\mathbb{E}}\!\!\left[\left|a_r\right|^2\!\right]}\!\! +  
    {\!\frac{ {{\mathscr{L}}^{\rm s}_{{\rm{cas}}}} \!N\!_{\rm D}}{{1\!+\!{\varepsilon}} }} {\sum\limits_{r = 1}^{{N_{\rm{D}}}}\!{\mathbb{E}}\!\!\left[\left|b_r\right|^2\right]}.
    \label{HDele}
\end{equation}

According to~\eqref{hDsRewirextenda} and~\eqref{hDsRewirextendb}, 
we have 
\begin{equation}
    {\mathbb{E}}\!\!\left[\left|a_r\right|^2\right] = {\mathbb{E}}\!\!\left[\left|b_r\right|^2\right] = \frac{{\mathbb{E}}\!\!\left[\left|{\beta_{r}\!(t)}\right|^2\right]}{ {N\!_{\rm D}}} 
    = \frac{\sum\limits_{i = 1}^{{2^b}} {{p_i}{{\mu^2_i}}}}{N\!_{\rm D}} = \frac{\overline \nu}{N\!_{\rm D}}.
    \label{HDele1}
\end{equation}
Therefore, \eqref{HDele} is then reduced to 
\begin{equation}
    {\mathbb{D}}\!\!\left[{{h}_{{\rm D},n}^{\rm s}}\!(t)\right] = 
    {{\mathscr{L}}^{\rm s}_{{\rm{cas}}}} \!N\!_{\rm D}{\overline \nu}.
    \label{HDele2}
\end{equation}

Consequently, the i.i.d. elements of the time-varying DRIS-based sensing path
${\boldsymbol{h}_{\rm D}^{\rm s}}\!\left(t\right)$ have mean zero and variance 
${{\mathscr{L}}^{\rm s}_{{\rm{cas}}}} \!N\!_{\rm D}{\overline \nu}$.
According to the Lindeberg-L$\acute{e}$vy central limit theorem,
the random variable ${{h}_{{\rm D},n}^{\rm s}}\!(t)$ converges in distribution to a normal distribution
as $N\!_{\rm D} \to \infty$, i.e., 
\begin{equation}
    {\boldsymbol{h}_{\rm D}^{\rm s}}\!(t)   
    \mathop  \to \limits^{\rm{d}}  \mathcal{CN}\!\!\left( {{\bf 0},  {{{\mathscr{L}}^{\rm s}_{{\rm {cas}}}}\!N\!_{\rm D}{\overline \nu} {\bf I}_{{N\!_{\rm B}}}} } \right),
    \label{hDsGua}
\end{equation}
where $\bf 0$ is a zero vector and ${\bf I}_{{N\!_{\rm B}}}$ is an $N\!_{\rm B} \times N\!_{\rm B}$ identity matrix.

\end{appendices}

\end{document}